\documentclass[11pt,a4paper]{article}
\usepackage[T1]{fontenc}
\usepackage[utf8]{inputenc}
\usepackage[english]{babel}

\usepackage{xargs}
\usepackage{xspace}
\usepackage{amsthm}
\usepackage{amsmath}
\usepackage{amssymb}
\usepackage{enumitem}
\usepackage{subcaption}
\usepackage[ruled, lined, linesnumbered]{algorithm2e}
\usepackage{tcolorbox}
\usepackage{multirow}
\usepackage{thmtools,thm-restate}
\usepackage{hyperref}
\usepackage{cleveref}
\usepackage{comment}
\usepackage{tikz}
    \usetikzlibrary{positioning}
    \usetikzlibrary{patterns}
\usepackage{subcaption}
\usepackage{macro/podc-macro}

\newtheorem{lemma}{Lemma}

\newtheorem{theorem}{Theorem}

\newtheorem{proposition}{Proposition}
\newtheorem{definition}{Definition}
\newtheorem{corollary}{Corollary}

\def\leq{\leqslant}\def\le{\leq}
\def\geq{\geqslant}\def\ge{\geq}
\def\emptyset{\varnothing}

\usepackage{xparse}
\NewDocumentCommand\set{sm}{\IfBooleanTF#1{\{{#2}\}}{\left\{{#2}\right\}}}
\NewDocumentCommand\ceil{sm}{\IfBooleanTF#1{\lceil{#2}\rceil}{\left\lceil{#2}\right\rceil}}
\NewDocumentCommand\floor{sm}{\IfBooleanTF#1{\lfloor{#2}\rfloor}{\left\lfloor{#2}\right\rfloor}}
\NewDocumentCommand\pare{sm}{\IfBooleanTF#1{({#2})}{\left({#2}\right)}}
\NewDocumentCommand\range{smm}{\IfBooleanTF#1{\set*{{#2},\dots,{#3}}}{\set{{#2},\dots,{#3}}}}
\NewDocumentCommand\card{sm}{\IfBooleanTF#1{|{#2}|}{\left|{#2}\right|}}

\title{Distributed Freeze Tag: a Sustainable Solution to Discover and Wake-up a Robot Swarm}
\author{Cyril Gavoille, Nicolas Hanusse, Gabriel Le Bouder, Taïssir Marcé}
\date{\vspace{-5ex}}

\begin{document}

\maketitle

\begin{abstract}
    The Freeze Tag Problem consists in waking up a swarm of robots starting with one initially awake robot. Whereas there is a wide literature of the centralized setting, where the location of the robots is known in advance, we focus in the distributed version where the location of the robots $\P$ are unknown, and where awake robots only detect other robots up to distance~$1$. Assuming that moving at distance $\delta$ takes a time $\delta$, we show that waking up of the whole swarm takes $O(\rho+\ell^2\log( \rho/\ell))$, where $\rho$ stands for the largest distance from the initial robot to any point of $\P$, and the $\ell$ is the connectivity threshold of $\P$. Moreover, the result is complemented by a matching lower bound in both parameters $\rho$ and $\ell$. We also provide other distributed algorithms, complemented with lower bounds, whenever each robot has a bounded amount of energy.
\end{abstract}

\section{Introduction}
In order to save energy in distributed systems, the paradigm of sleeping models and algorithms has received a recent attention. Nodes or robots are, by default, inactive or on standby: the energy consumption is negligible and these periods can be used to harvest energy. A robot becomes active only if it is required.

The Freeze Tag Problem (FTP) consists in waking up a swarm of $n$ inactive (or sleeping) robots as fast as possible assuming that one robot is initially active. 
To be woke up, a sleeping robot has to be reached by an awake robot, that are able to move in the plane.
Once a robot becomes active, it can help wake up other robots.

FTP has been introduced in a centralized setting, where the $n$ locations of the sleeping robots are known by the initial awake robot $s$. In this article, we propose a distributed version of the FTP: (1) the locations of the sleeping robots are not known in advance; (2) using a local snapshot, active robots only have a distance-1 visibility; and (3) robots need to be co-located to communicate. Note that due to the visibility constraint, it may be required to explore further than radius~1 to locate sleeping robots.

In order to get the most sustainable solution in a long-life perspective, we aim at minimizing the energy consumption. In particular, since moving is identified as an energy intensive task, the goal is to minimize \emph{the makespan}, that is the time to wake up every robot, assuming unitary speed of the robots, i.e., moving at distance $\delta$ takes $\delta$ unit of time. It is also assumed that robots use discrete snapshots only, since continuous snapshots may by not energy friendly.
We observe that since in our distributed model, robots do not know in general the locations of other robots, robot $s$ has to move at least $\Omega(D^2)$ to discover and reach the closest other robot if it is located at distance $D$ from $s$. Obviously, it can reach it with time $O(D^2)$ by following the trajectory of a spiral starting from $s$ for instance.

\subsection{State of the art}

Any solution to the FTP can be seen as a rooted tree spanning the a set of $n+1$ robot's positions, called the \emph{wake-up tree}. The root node corresponds to the position of~$s$, has one child, and the $n$ node are other robots have at most two children each. Each edge has a \emph{length},
representing the distance in the metric space between its endpoints. The makespan of the solution is nothing else than the (weighted) depth is the wake-up tree, and in particular a solution with optimal makespan has a wake-up tree with minimum depth.

\paragraph{Freeze Tag.}

Even for simple case, an optimal solution of the FTP can not be computed in polynomial time. Arkin \textit{et al.}~\cite{ABFMS06} showed that, even on star metrics, FTP is NP-hard. Moreover, they proved that getting an $5/3$-approximation is NP-hard for general metrics on weighted graphs. However, in~\cite{ABGHM03}, the authors give a polynomial time algorithm to get an $O(1)$ for general graphs, assuming one sleeping robot per node.

In this paper, we focus on the \emph{geometric} setting, where robot movements have no restrictions and where the position set $\P$ lie on the Euclidean plane. Even in this setting, the problem remains NP-hard~\cite{AAJ17}. It has been shown by Yazdi \textit{et al.}~\cite{YBMK15} that a wake-up tree of makespan of at most~$10.07\rho$ can be (sequentially) computed in time $O(n)$, where $\rho$ is the largest distance from $s$ to any point of $\P$. The constant $10.1$, aka the \emph{wake-up constant} of the Euclidean plane, has been later on improved by Bonichon \textit{et al.}~\cite{BCGH24} to~$7.07$. More generally, they proved that the wake-up constant for \emph{any} norm is no more than~$9.48$, and that a corresponding wake-up tree can be computed in time~$O(n)$. Very recently, the upper bound dropped to independently to $5.06$ by~\cite{ABMS25} and to $4.63$~\cite{BGHO24}. It is known that $1+2\sqrt{2} \approx 3.83$ is a lower bound on the wake-up constant of the plane~\cite{BCGH24}.

A first step toward the computation of a wake-up tree without a global knowledge of the robot's positions is the \emph{on-line setting}~\cite{HNP06,BW20}. In this case, each robot only appears at a specified time that is not known in advance.
In~\cite{BW20}, the authors propose a solution with a competitive ratio of $1+\sqrt{2}$ w.r.t. to the optimal partial wake-up tree. 

\paragraph{Collaborative Exploration.}

Obviously, any collaborative exploration problem requires to have a team of active robots. Conversely, the \emph{distributed FTP} (dFTP for short), requires to explore some area to discover sleeping robots and thus is naturally connected to exploration of the plane with one or more robots. The survey of Das~\cite{Das19} contains many references such exploring problems in unknown graphs. In dFTP, we start with one active robot, and after few steps, we can have $k$ active robots. So, the task of discovering new robots can indeed be seen as a collaborative exploration task. The question of improving exploration with the use of $k>1$ robots is challenging and widely open.
For instance, it has been shown by Fraigniaud \textit{et al.}~\cite{FGKP06} that unweighted trees of diameter $D$, distributed exploration can be done in $O(D + n/ \log{k})$ unitary moves, even if robots are allowed to let some information at the nodes, whereas we could hope a speed-up of $k$ with $O(D + n/k)$ unitary moves. However, if the underlying graph is a two dimensional sub-grid of $n$ vertices, a grid with rectangular holes, Ortlof and Schindelhauer~\cite{Ortolf12} show how to get an optimal speedup of factor~$k$.

Discovering a robot at distance $D$ with $k$ co-located robots in the plane can be done within $\Theta(D+D^2/k)$ unitary moves per robot using either parallel spiral trajectories~\cite{Fricke16}, or by partitionning into a square of width $D$ into $k$ rectangles of width $D/k$ and height $D$. This problem, aka Treasure Hunt Problem or Cow-Path Problem, has been widely studied for $k = 1$ or with imprecise geometry~\cite{Bouchard20}. Interestingly, the authors of~\cite{FKLS12} have showed that the knowledge of an approximation of $k$ is required to get a time the bound $\Theta(D+D^2/k)$. The question of the knowledge of $k$ arises whenever the $k$ robots do not start the exploration together (as in the dFTP), or whenever the communication ability of the robots is limited. 

\paragraph{Energy Consumption.}

Some recent works deal with the problem of distributed tasks with some energy constraints or minimizing the energy consumption. In the sleeping model~\cite{Chatterjee20}, nodes or robots are either sleeping or are active. If a robot is sleeping, its consumption is assumed to be negligible. The state of each robot in synchronized rounds is given by a centralized schedule. Distributed tasks have been considered in the sleeping model like coloring or MIS computations. 

Energy Constrained Exploration Problems are perhaps more related to our problem. For instance, in the \emph{Piece-Meal Graph Exploration}, robots have a given budget for the energy and needs to refuel at a home base before exploring unknown parts of the graph. Note that a solution of the treasure hunt in a grid graph can be used for the plane. In~\cite{Duncan06}, the authors show that $n$-node and $m$-edge unweighted graphs of radius $R$ can be explored by $k = 1$ agent with an energy budget $B = (1+\alpha)\cdot R$ in $O(m + n/\alpha)$ unitary moves. For $k>1$,~\cite{Dynia06,Das24} deal with the Energy Constrained Depth First Search while minimizing the number $k$ of robots with an energy budget $O(R)$ per robot to explore a tree of radius $R$. Other distributed algorithms for energy constrained agents has been considered in~\cite{Bartschi20}. They show how to provide a feasible movement schedule for mobile agents for the Delivery Problem, where each agent has limited energy which constrains the distance it can move. Hence multiple agents need to collaborate to move and deliver the package, each agent handing over the package to the next agent to carry it forward. However, the positions of the agents are assumed to be known and the computation is centralized.

However, all these results related to collaborative exploration and energy consumption are not directly related to our setting, and essentially because we are face to the fact that, by definition of the problem, the number of active robots collaborating keeps on evolving, from~$1$ to~$n+1$.
\subsection{The Model}

\paragraph{Computational Model.}

We consider a swarm of robots in the Euclidian plane. Robots are all initially asleep, except one which we call the \emph{source} and denote $s$, initially located at position $p_0 = (0,0)$. The set of all robots is denoted by $\R = \{s, \rob_1, \dots, \rob_n\}$, robots $r_i$ being the initially asleep robots. We denote by $p_i$ the initial position of robot $\rob_i$, and by $\P$ the set of all initial positions of initially asleep robots, i.e., $\P = \range{p_1}{p_n}$. Robots are endowed with a visible light indicating their status (sleeping or awake), which can be observed by any active robot close enough (in its distance-1 vicinity). Sleeping robots are computationally inactive. They can neither move, observe, nor do any type of computation. Awake robots are aware of the absolute coordinate system, a same global clock and are able to locate and distinguish sleeping and awake robots in their vicinity, by using a function \look. They can also share variables of their memory with co-localated robots, and can operate computations based on the information they gathered previously. Finally, they can move in the plane, based on the computation they operated. Robots move at speed~$1$, which means it takes a time $\delta$ for a robot to move between any two points at Euclidean distance $\delta$. Thus the behaviour of a robot can be described in the standard Look-Compute-Move Model (see \cite{Flocchini19}).For synchronization purpose, robots can also wait for any duration of time at a fixed position. When an awake robot and a sleeping robot are co-located, the awake one can wake the other one up, and possibly share with it some information as previously said.

An algorithm $\A$ aiming to solve the dFTP is executed in parallel by all the awake robots.
The execution of $\A$ terminates when all active robots have terminated their computation and moves. The execution is valid if, when it terminates, all the initially asleep robots have been awakened. The makespan of an execution is the duration between the beginning of the algorithm and its termination, which basically counts the duration of the moving and waiting actions of robots. 

Robots have a local unlimited memory. Typically, they can store the positions of some robots and their status (sleeping/awake) at the time they see them in their vicinity. Note that robots can give themselves a globally unique identifier as soon as they are awakened, by storing their initial position.
We will also consider the variant of the model where the robots are not free to move as long as they want, but are rather limited by some \emph{energy budget} $B$. In this variant, a robot can move for a total distance at most $B$.

\paragraph{Spread of $\P$.}

Our results are highly linked to the distribution of $\P$. Typically, if the distance between every pair of robots is much larger than~$1$, a robot may have the inaccurate belief that it is alone\footnote{The co-located robot that activated it excepted.}
which complicates a lot the resolution of the problem. To formally present our results as in Table~\ref{tab:contributions}, let us introduce some parameters related to $\P \subset \mathbb{R}^2$.

Given a real $\delta\ge 0$, and $\X \subset \mathbb{R}^2$, the \emph{$\delta$-disk graph} of $\X$ is the edge-weighted geometric graph whose vertex set is $\X$, two points $u,v \in \X$ being connected by an edge if and only if $u$ and $v$ are at (Euclidean) distance at most~$\delta$, and the weight of the edge corresponds to the distance between their endpoints.

Let $(\P,s)$ be an $n$-point set $\P\subset \mathbb{R}^2$ with a source $s \notin \P$.
The \emph{radius} of $(\P,s)$, denoted by $\rhostar$, is the largest distance from $s$ to any point of $\P$. 
The \emph{connectivity-threshold} of $(\P,s)$, denoted by $\ellstar$, is the least radius~$\delta$ such that the $\delta$-disk graph of $\P\cup\set{s}$ is connected. 
Given $\ell>0$, the \emph{$\ell$-eccentricity} of $(\P,s)$, denoted by $\ecc_\ell$, is the -- finite or infinite -- minimum weighted-depth of a spanning tree of the $\ell$-disk graph of $\P\cup\set{s}$ rooted at $s$. 

\paragraph{Problem Definition.}

Note that if $\rhostar \leq 1$, every robot is seen by the source $s$ and can be waken up in time $O(1)$ with energy budget $O(1)$ by solving the centralized version in $s$, e.g., as done in~\cite{BCGH24}. 
We shall suppose that $s$ starts with some information about the connectivity-threshold, the radius, and the number of asleep robots in $\P$ that it is supposed to wake up. More precisely, a tuple of values $(\ell, \rho, n)$ is given to $s$ at its start. Indeed, without any information, it is not difficult to see that $s$ cannot terminate (and thus cannot solve the dFTP), being unable to distinguished (for instance) the case where $n=0$ ($s$ is alone) from $n>0$, without moving for eternity. 
An algorithm with input $(\ell, \rho, n)$ and solving dFTP should terminate on any $n$-point set $(\P,s)$ such that $\ellstar \leq \ell$ and $\rhostar \leq \rho$.
Also note that we always have $\ellstar \leq \rhostar \leq n \ellstar$, as proven in Proposition~\ref{prop:admissible}. And, so a tuple $(\ell,\rho,n)$ is said \emph{admissible} if $\ell \le \rho \le n \ell$.
Also, for the sake of simplicity, we suppose that parameters $\ell$ and $\rho$ are positive integers.
This hypothesis does not actually change the problem (in term of asymptotic complexity of the makespan), since $(\ell, \rho, n)$ is admissible if and only if $(\ceil{\ell}, \ceil{\rho}, n)$ is.

The dFTP is formally defined as follows:
\begin{definition}[dFTP]
    A distributed algorithm $\mathcal A$ solves the dFTP if, for any admissible tuple $(\ell, \rho, n)$, and for any $n$-point set $\P$ with source $s$ such that $\rhostar \leq \rho$ and $\ellstar \leq \ell$, the execution of $\mathcal{A}$ on $\P$ at source $s$, given $(\ell, \rho, n)$, eventually wakes up all the robots and terminates. Moreover, it solves the dFTP with energy budget $B$ if the previous holds, assuming that the total movement lengths of each robot does not exceed $B$.
\end{definition}

\subsection{Contributions}

Our contributions are summarizes in~\Cref{tab:contributions}. 
We present three algorithms, called \Aseparator, \Agrid and \Awave.

The first algorithm \Aseparator solves dFTP, with no limits on the energy budget, and has makespan $O(\rho + \ell^2 \log(\rho/\ell))$. This result is complemented by a matching lower bound.

The two other algorithms consider the dFTP with energy budget $B$. We first show that no algorithm can solve the limited energy budget variant if $B < c\ell^2$, for some constant $c>0$. Then, for $B \in \Theta(\ell^2)$, \ie as little energy as possible to solve the dFTP, \Agrid achieves a makespan of $O(\ecc_\ell\cdot\ell)$. Using slightly more energy, namely $B \in \Theta(\ell^2\log{\ell})$, \Awave has a significantly lower makespan, which matches a second lower bound we introduce.

\begin{table}[htbp!]
    \centering
    \renewcommand{\arraystretch}{1.1}
    \begin{tabular}{c|c|c||c}
        Energy & Algorithm &  Makespan & Lower Bound\\
        \hline\hline
        unconstrained & \Aseparator & $O(\rho + \ell^2 \clog{(\rho/\ell)})$ - Th. \ref{th:upper} & $\Omega(\rho + \ell^2 \clog{(\rho/\ell)})$ - Th.~\ref{th:lower} \\\hline
        $< \pi(\ell^2-1)/2$ & - & - & unfeasible - Th.~\ref{th:impossibility}\\\hline
        $\Theta(\ell^2)$ & \Agrid & $O(\ecc_\ell \cdot \ell)$ - Th.~\ref{th:Agridmakespan} & \multirow{2}{*}{ $\Omega(\ecc_\ell + \ell^2 \log{(\ecc_\ell/\ell)})$ - Th.~\ref{th:lower-nrj}}\\\cline{1-3}
        $\Theta(\ell^2\log \ell)$ & \Awave & $O(\ecc_\ell + \ell^2 \log({\ecc_\ell / \ell}))$  - Th.~\ref{th:Awavemakespan} & \\\hline
    \end{tabular}\\[.25ex]
    \caption{Complexity of the makespan for the dFTP given $(\ell,\rho,n)$.}
    \label{tab:contributions}
\end{table}

\paragraph{Roadmap.}
In Section~\ref{sec:building-blocks} we present the main building blocks of our algorithm.
These are high-level procedures we use to describe \Aseparator in Section~\ref{sec:without-energy}. In Section~\ref{sec:energy} \Aseparator will also be used as a building block for algorithms with constrained energy \Agrid and \Awave.
Due to space limitations, the majority of the descriptions of our algorithms and proofs is provided in a separate appendix. 

\section{Building Blocks}
\label{sec:building-blocks}

The main parts of our algorithms are based on exploring regions, computing and realizing wake-up trees, and organizing teams of robots to explore regions in parallel. The recruitment of a team is based on a sampling of point sets. To avoid exploring large empty regions, we use geometric separators.

\subsection{Exploration (\explore)}
\label{sec:section-explore}

One central task robots are led to realize is the exploration of a given region, in order to collect the positions of all the sleeping robots in that region. For the sake of simplicity, we only consider rectangular regions, whose orientation is parallel to the axis.
Note that it can be used to explore any shape inscribed in a rectangle.
We present in~\Cref{sec:explore-details} a simple method for exploring a given rectangle with a single robot, using function \look, which can be adapted to a team of robots. In this extension, every robot explores a sub-rectangle before moving to a meeting point where they can share their variables.
\begin{restatable}[\explore]{lemma}{rstexplore}
    \label{lem:collaborative-exploration}
    There exists a procedure $\explore$ such that, 
    for any rectangle $\squarereg$ of dimensions $w \times h$, for any two positions $p,p' \in \squarereg$,
    the execution of $\explore(\squarereg,p')$ at time $t$, by a team of $k$ robots $\team = \{r_1,\dots,r_k\}$ initially co-located at position $p$ guarantees:
    \begin{itemize}
        \item it terminates at time $t'$ with $(t' - t) \in O(wh/k + w + h)$; and
        \item at time $t'$, robots of $\team$ have gathered the initial positions of all robots of $\squarereg$ that are asleep at $t'$.
    \end{itemize}
\end{restatable}

\subsection{Realization of a Central Wake-up-Tree}
\label{sec:wakeup-propagation}

In~\cite{YBMK15,BCGH24} the authors show that, knowing the initial positions of robots, it is possible to compute a wake-up tree in linear time, whose makespan is an approximation of the optimal. Yet, in the distributed setting, some specific problems may arise. Indeed, two awake robots $\rob_i$ and $\rob_j$ may compute independently two wake-up trees on different but not disjoint subsets $X_i$ and $X_j$ of $\P$. If $\pos[k]$, the position of $\rob_k$, belongs to $X_i \cap X_j$, then $\rob_i$ and $\rob_j$ are said \emph{in conflict}. 
Both need to use $\rob_k$ for their wake-up trees and only the first robot to reach $\rob_k$ is able to do it (since $\rob_k$ will then move). The second one only find out that $\rob_k$ has left its initial position $\pos[k]$ when itself or one of its descendent arrives close to $\pos[k]$. In this situation the consistency of the wake-up tree is broken, and a new wake-up tree has to be computed for the sub-tree rooted at $\rob_k$. Since this situation can be repeated an arbitrarily large number of times for a set $Y = X_i \cap X_j$, there is no evidence that $\rob_j$ can wake up $X_j \setminus Y$ in a time proportional to the diameter of $X_j \setminus Y$. To deal with that issue, we ensure in our algorithms that wake-up trees are computed in separate regions of the plane, and that at most one robot computes a wake-up tree in a given region. This is formalized in~\Cref{lem:centralized-awakening}, detailed in~\Cref{sec:wut-realization}.

\begin{restatable}[Distributed Makespan]{lemma}{centralized}
\label{lem:centralized-awakening}
    Given a square region $\squarereg$ of width $R$ and a robot $\rob$ positioned in the center of $\squarereg$, knowing both $\squarereg$ and a set of sleeping robots $\Sleeping$ whose initial positions are in $\squarereg$,
    if no robots other than $\set{\rob}\cup\Sleeping$ takes action in $\squarereg$,
    then $\rob$ can wake-up all robots of $\Sleeping$ in time $5R$.
\end{restatable}

More generally, if robots do not know the positions of sleeping robots in $\squarereg$, it is possible for a robot to discover them before computing a wake-up tree. This process will be used in \Agrid, and is presented in the following Corollary, detailed in Section~\ref{sec:local-sync}.

\begin{restatable}[Explore and Wake up]{corollary}{simplesquarewup}
    \label{cor:simple-square-wup}
    Given a square region $\squarereg$ of width $R$ containing at least one awake robot, if no awake robot goes through the border $\squarereg$, it is possible to wake-up all sleeping robots of $\squarereg$ in time $R^2 + (10 + \sqrt{2})R$.
\end{restatable}

\subsection{Geometric separators}

One central tool we use is \emph{geometric separators}.
Given a square $\squarereg$ of width $R$ centered at position $p$, the interior of $\squarereg$, including $\squarereg$, is noted $\squarereg^{\mathit{in}}$ and the remaining region of the plane is noted $\squarereg^{\mathit{out}}$.
Given a square $\squarereg$ with width $R > 2\ell$, we define the \emph{separator of $\squarereg$}, denoted $\sep(\squarereg)$, as the region bounded by $\squarereg$ and a square of center $p$ and width $R-2\ell$.
We immediately get:
\begin{lemma}
    \label{lem:separator-path}
    Let $\squarereg$ be a square.
    Let $\ell$ be the connectivity threshold of $(\P,s)$.
    Any path in the $\ell$-disk graph of $\P$ linking robots $\rob \in \squarereg^{\mathit{in}}$ and $\rob' \in \squarereg^{\mathit{out}}$ contains at least one robot located in $\sep(\squarereg)$.
\end{lemma}
\begin{corollary}
    \label{cor:separator}
    If $\P \cap \sep(\squarereg) = \varnothing$, either $\P \subset \squarereg^{in}$ or $\P \subset \squarereg^{out}$.
\end{corollary}

\subsection{Distributed \texorpdfstring{$\ell$}{l}-sampling}
\label{sec:dfsampling}
To begin, we point out that the time to wake up every robot in a square region of width $R$, $O(R)$ (Lemma \ref{lem:centralized-awakening}), is often dominated by the time to discover sleeping robots in this region by a team of $k$ robots, which is $O(R^2/k)$ (Lemma~\ref{lem:collaborative-exploration}).
However, knowing an upper bound $\ell$ on the connectivity threshold can help thanks to \emph{$\ell$-samplings}.  
More formally, \emph{an $\ell$-sampling of a region $\squarereg$ is a subset of positions $\P' \subseteq \P \cap \squarereg$ that are pairwise at distance at least $\ell$}. 
We also say that $\squarereg$ is \emph{covered} by $\P'$, if any robot within $\squarereg$ is at distance at most $\ell$ from a robot whose position is in $\P'$.
Assuming that the $\ell$-sampling is given and $\P$ is covered by $\P'$, discovering the $n$ robots can be done in parallel in time $O(\ell^2)$, gathering $\P'$ at the source $s$ in time $R$ and waking-up every robot in time $O(R)$ using a centralized wake-up algorithm (Lemma \ref{lem:centralized-awakening}). 
In total, it takes  $O(R+\ell^2)$. 
The two main difficulties are how to efficiently compute in a distributive way an $\ell$-sampling, and how large is an $\ell$-sampling.

\begin{restatable}[$\ell$-sampling cardinality]{lemma}{rstsamplingsize}
\label{lem:disk-covering}
If $\P'$ is an $\ell$-sampling of a square region of width $R$, then $|\P'| \leq 16 R^2 / (\pi \ell^2)$.
\end{restatable}

\dfsampling is a distributed algorithm computing an $\ell$-sampling of a squared region $\squarereg$. 
A single robot or a co-located team of robots start from a same position of $\P \cap \squarereg$ and aims at finding an $\ell$-sampling of size $4 \ell$ (See  Figure \ref{fig:DFS-sep}). 
This sampling represents positions of robots to be later recruited to become a new team. 
This sampling is discovered using a specific exploration algorithm.
Since it may happen that exploring from one single position is not enough to reach a sample of size $4 \ell$, the team may use several starting positions for the exploration task, which we call the \emph{seeds} $\X \subset \squarereg$.
In Figure~\ref{fig:exploration-sep}, $\X$ is the set of points within an $\ell$-separator of a sub-square.

Roughly speaking, \dfsampling is based on a Depth-First Search in the $2\ell$-disk graph of $\P \cap \squarereg$, starting by the position of seeds from $\X$. 
Whenever a point $p$ is discovered, it is added to $\P'$ only if it is at distance greater than $\ell$ from any other points already added to $\P'$. 
This is required to guarantee that $\P'$ is indeed a $\ell$-sampling. 
A detailed description is in Section~\ref{sec:detail-dfs} as well as the proof of Lemma~\ref{lem:dfsampling}.

\begin{restatable}[\dfsampling]{lemma}{rstdfsampling}
    \label{lem:dfsampling}
    In a region $\squarereg$ of width $R$ containing a set of seeds $\X \subseteq \P$, a team $\team$ of robots knowing the awake robots of $\squarereg$ can compute a set of positions $\P'$ being an $\ell$-sampling of $\P$. The computation is done using \dfsampling in time:
    \begin{align*}
        O(\ell^2 \clog(|\P'|)) & \text{ if } |\team|=1 \text{ and }  \X = \set{p_s} \text{ and } R\geq 2\rhostar\\
        O(R+\ell |\P'|) & \text{ if } |\team|=\ell \text{ and } \X = \P \cap \sep(\squarereg)
    \end{align*}
   In both cases, either (1) $|\P'|=4\ell$; or (2) $|\P'| < 4\ell$ and $\squarereg$ is \emph{covered} by $\P'$.
\end{restatable}

\section{Without energy constraint}\label{sec:without-energy}
\begin{figure}[h!]
    \captionsetup[subfigure]{position=b}
    \centering
    \subcaptionbox{\textsf{Initialization} \label{fig:initialization-sep}}{\includegraphics[width=.3\linewidth]{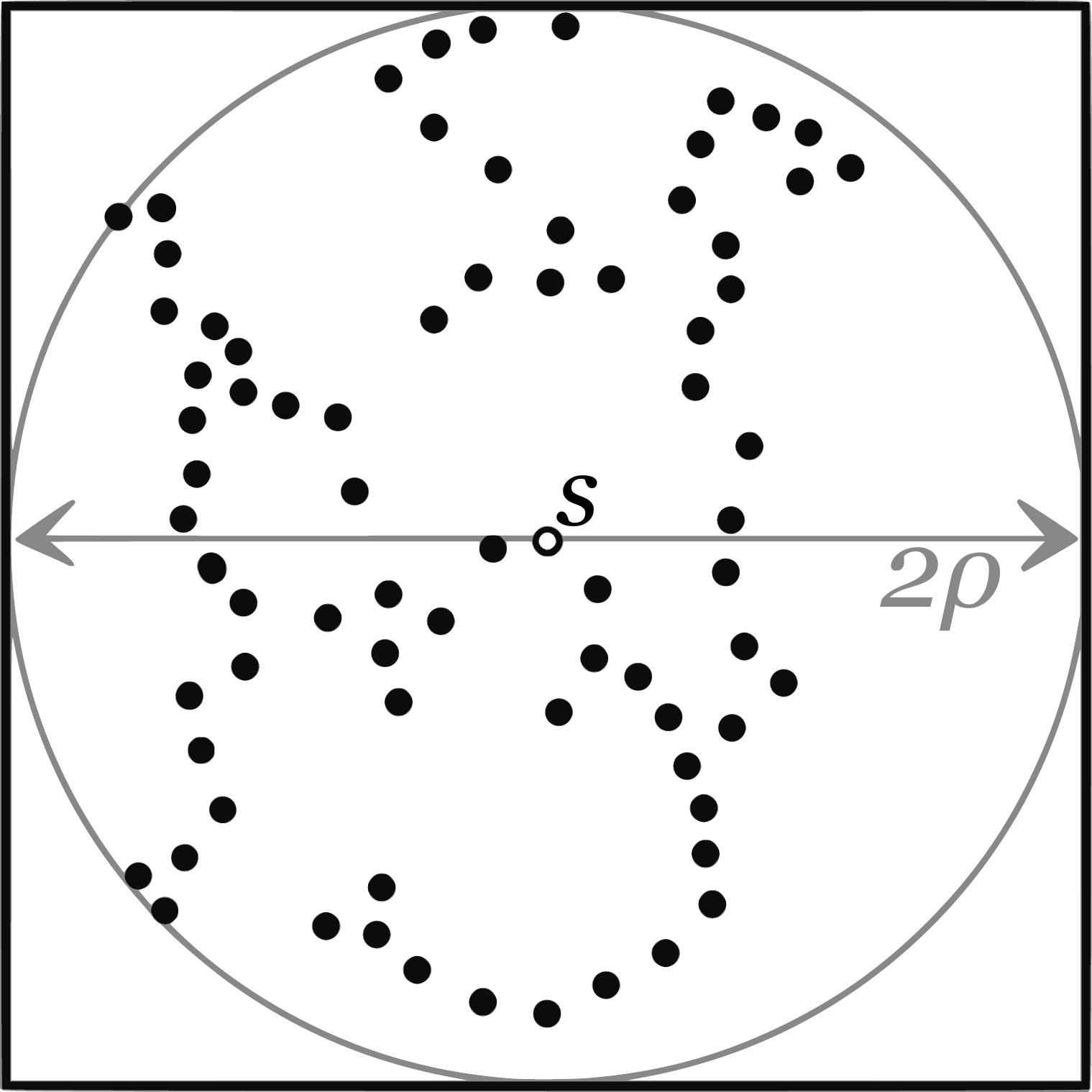}}
    \hfill
    \subcaptionbox{\dfsampling from $s$ to get a sample of size $4\ell$.\label{fig:DFS-sep}}{\includegraphics[width=.3\linewidth]{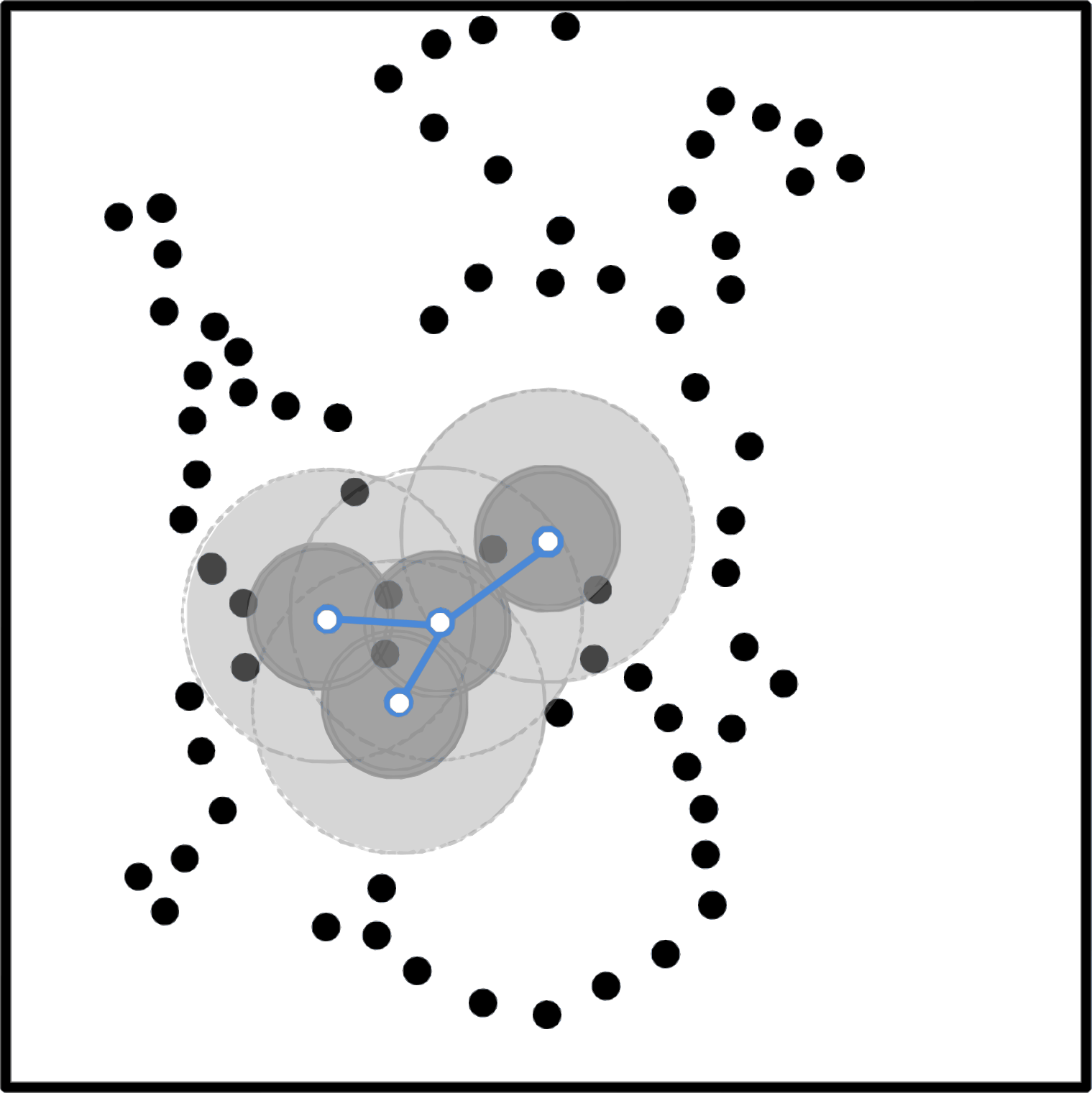}}
    \hfill
    \subcaptionbox{\textsf{Exploration} of $4$ separators by teams of $\ell$ robots. \label{fig:exploration-sep}}{\includegraphics[width=.305\linewidth]{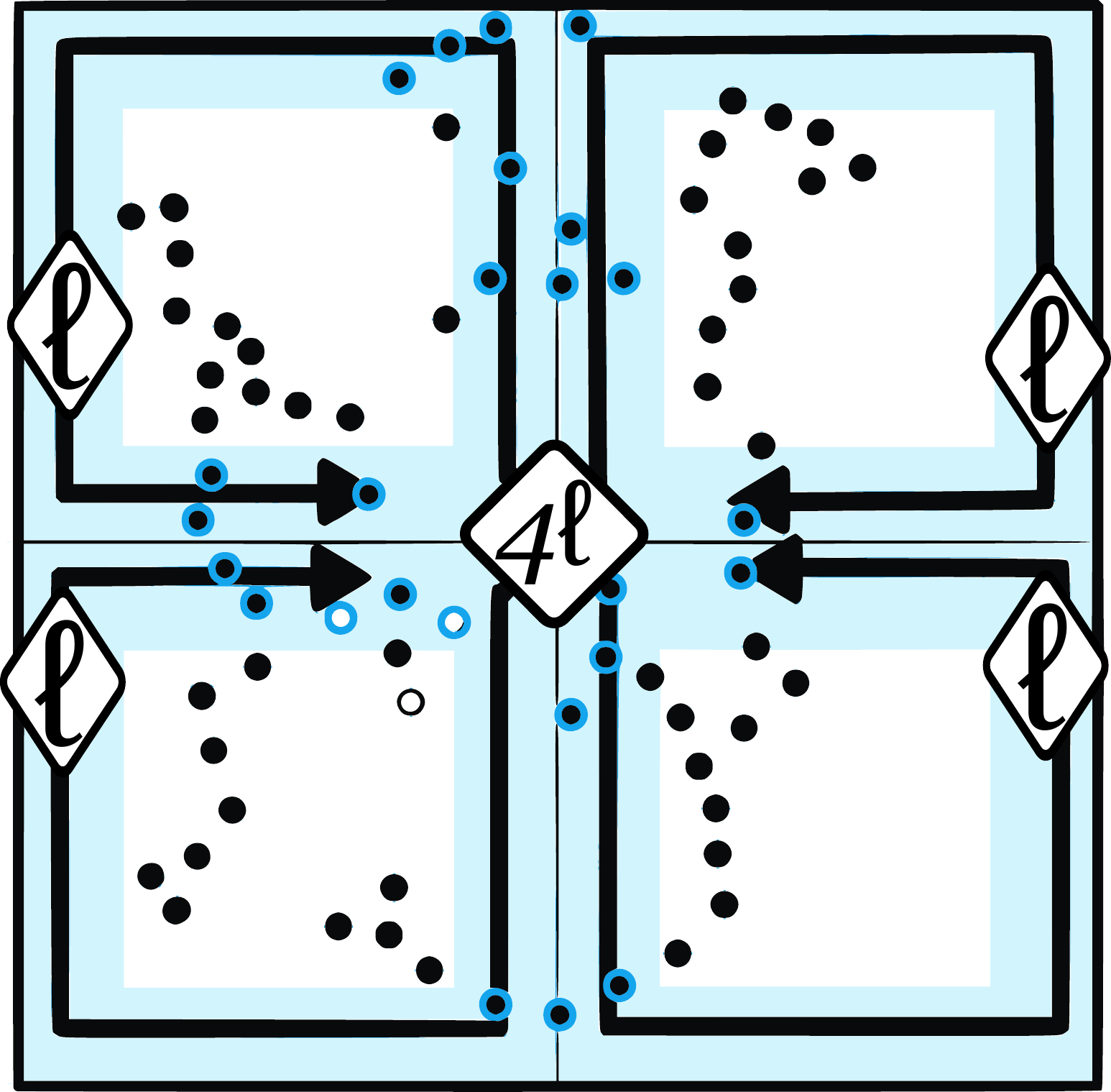}}
    \caption[First phases of \Aseparator]
        {First phases of \Aseparator; $\bullet$ sleeping robots; $\circ$ awake robots; \tikzbullet{blue-source}{black} sleeping seeds; \tikzbullet{blue-source}{white} awake seeds.}
\end{figure}

\begin{figure}[h!]
    \captionsetup[subfigure]{position=b}
    \centering
    
    \subcaptionbox{\textsf{Recruitment}: \dfsampling starting from $\X_i \subset \sep(\squarereg[i])$.
        \label{fig:parallel-DFS-sep}}{\includegraphics[width=.3\linewidth]{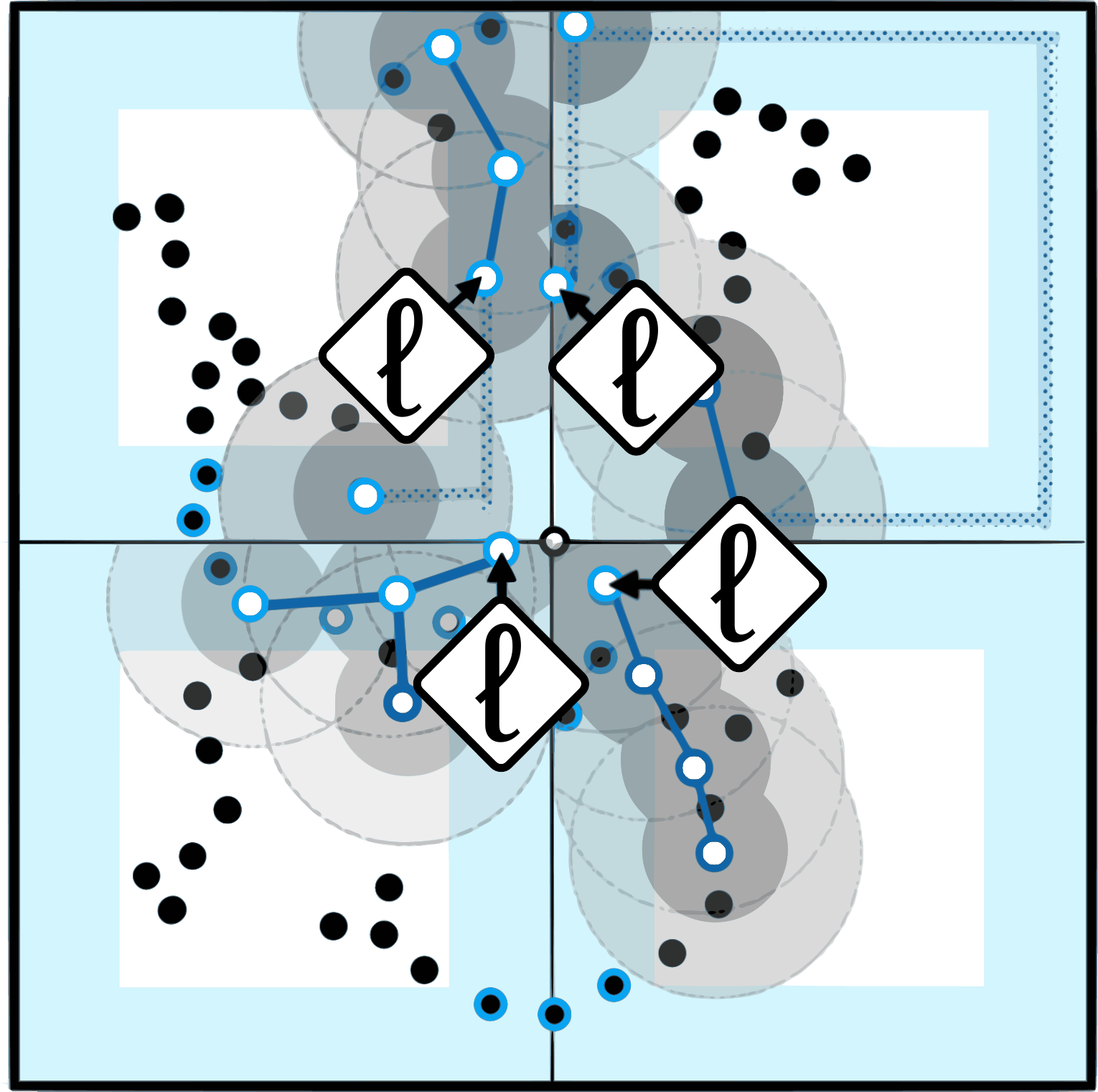}}
    \hfill
    \subcaptionbox{ $\team_i$ merge at the center of $\squarereg$ and share their variables.
        \label{fig:backhome-sep}}{\includegraphics[width=.3\linewidth]{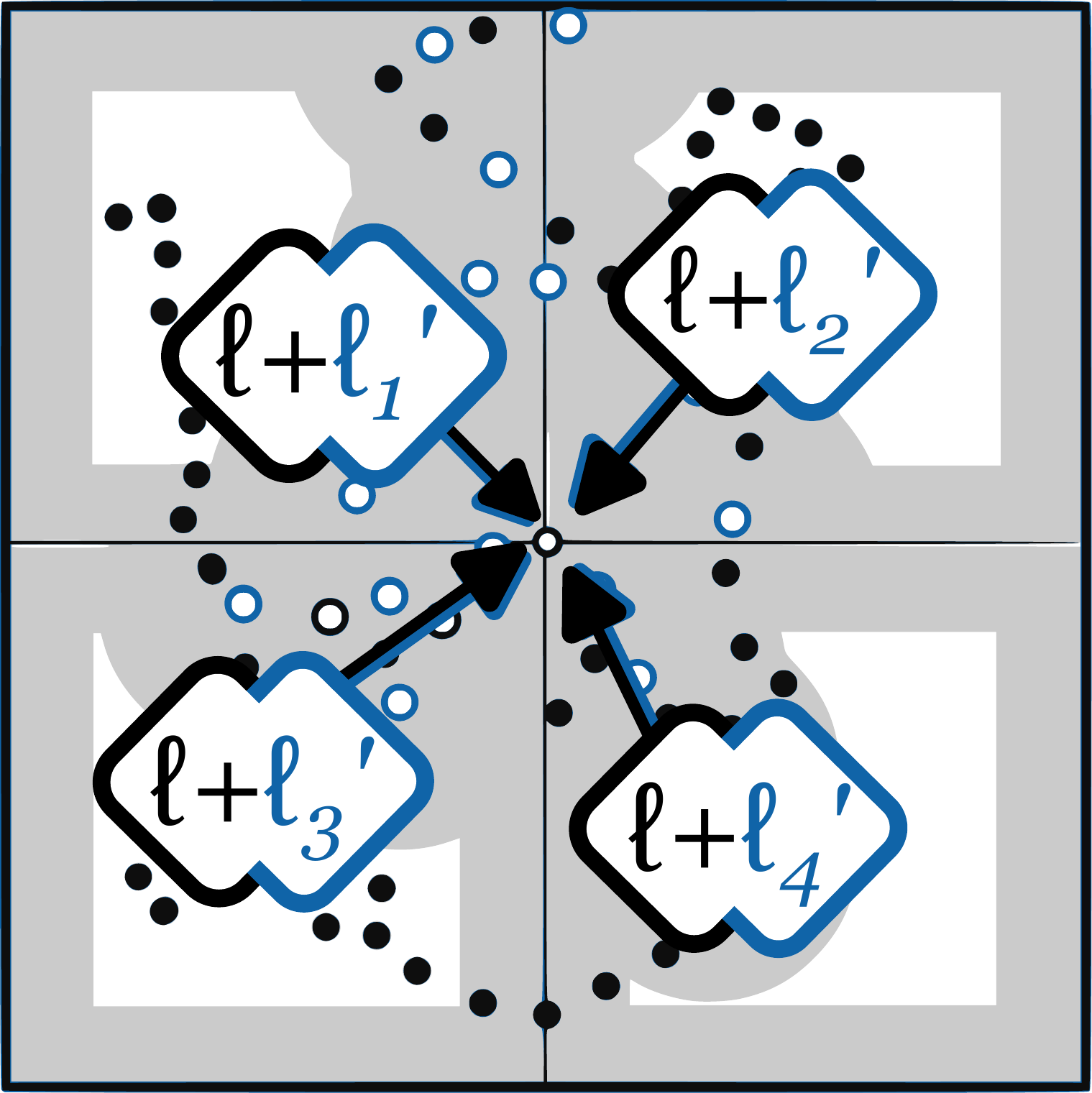}}
    \hfill
    \subcaptionbox{Next round: four team of $4\ell$ move to a sub-square.\label{fig:rec-sep}}{\includegraphics[width=.305\linewidth]{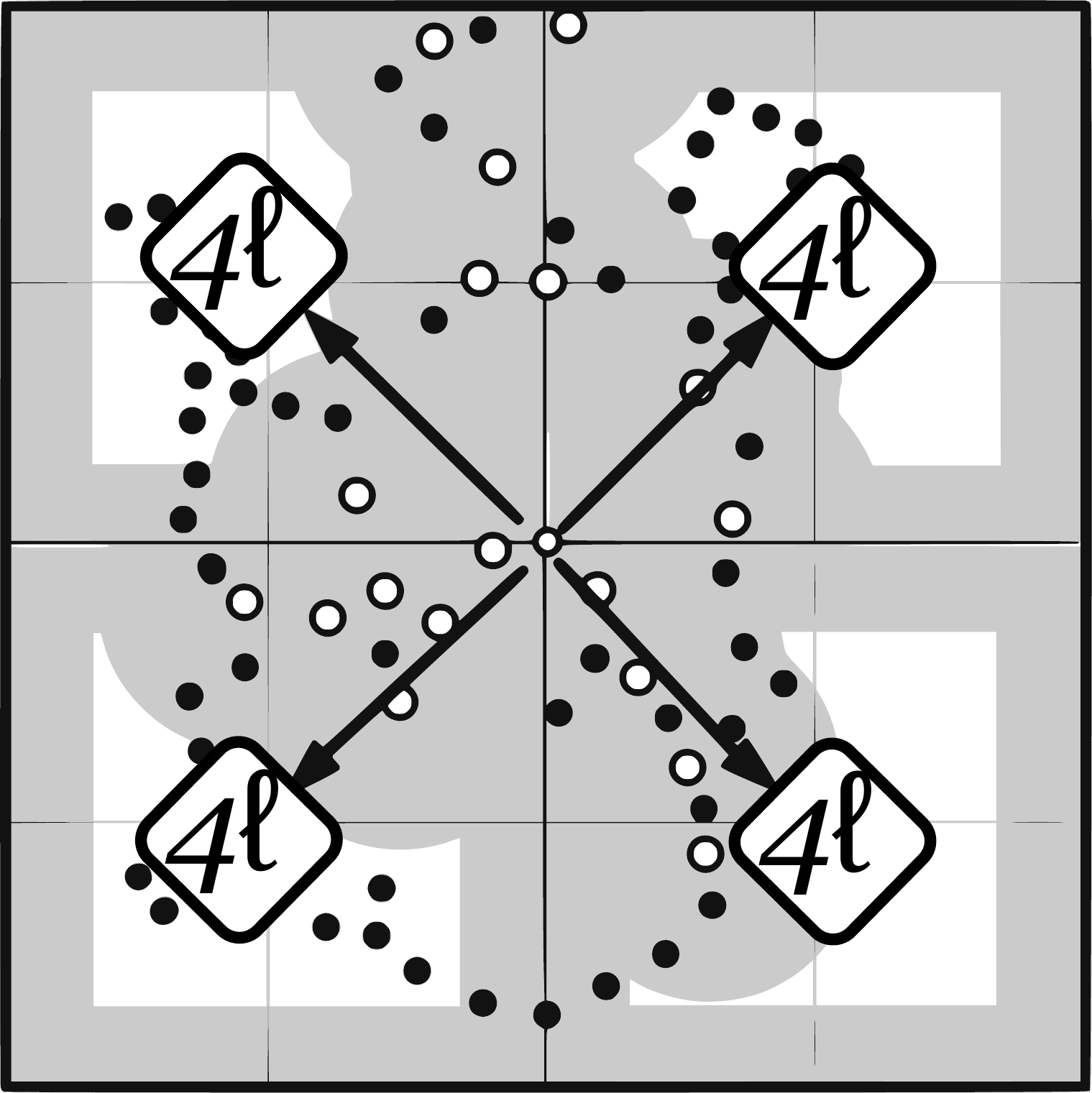}}
    \caption[Recruitment and Reorganization phases of Aseparator.]
        {\textsf{Recruitment} and \textsf{Reorganization} phases of \Aseparator; \tikzbullet{blue-recruits}{white} recruited robots; explored area is grayed.}
\end{figure}

Algorithm \Aseparator is based on a divide-and-conquer strategy where robots are organized in teams. As shown by a matching lower bound, the algorithm has optimal makespan. We describe it with six phases: \textsf{Initialization, Partition, Exploration, Recruitment, Reorganization} and \textsf{Termination}.

Having first awaken $4\ell$ robots in the square of center $s$ and width $2\rho$ (\textsf{Initialization} and \textsf{Recruitment} Phase, presented in Figure~\ref{fig:initialization-sep} and \ref{fig:DFS-sep}).
we divide the square into $4$ sub-squares, and send a team of $\ell$ robots in every sub-square (\textsf{Partition} Phase).
During the \textsf{Exploration} Phase, each team collectively explores the separator of their sub-square and stores seeds in \X, that are the initial positions of robots (sleeping or already awake) belonging to the separator (Figure~\ref{fig:exploration-sep}).
During the \textsf{Recruitment} Phase, each team wakes up new robots in its sub-square, such that these new robots \emph{plus robots currently exploring} whose initial position is in that sub-square, are $4\ell$   (Figure~\ref{fig:parallel-DFS-sep}). These $4 \ell$ robots are \emph{recruited} to define a new team used in the next round.

Finally, during phase \textsf{Reorganization}, all $4$ teams meet in the center of the square so robots can team up with robots initially located to the same sub-square. Each team thus formed go into its corresponding sub-square, and repeat the process until all robots have been awakened  (Figures~\ref{fig:backhome-sep}, \ref{fig:rec-sep}).
\textsf{Recruitment} Phase relies on function $\dfsampling$ presented in Section~\ref{sec:dfsampling}.
Lemma~\ref{lem:dfsampling} guarantee that if \dfsampling returns a set with less than $4\ell$ positions, then every robot located in the sampled region has been discovered (but not necessarily awakened).
This justifies the \textsf{Termination} phase where awaken robots wake up the remaining sleeping robots with a centralized algorithm.
Figure~\ref{fig:code-aseparator} presents a more detailed description of \Aseparator.

\begin{figure}[h!]
    \begin{tcolorbox}[sharp corners, colback=blue!8, colframe=blue!30!gray, title=\Aseparator]
        \begin{enumerate}
            \item \textbf{Round $0$}: \textsf{\textbf{Initialization}} and \textsf{\textbf{Recruitment}} \\
             $\squarereg \leftarrow $ square of width $R = 2\rho$ and centered at the source robot $s$ \\
            $\team \gets \{s\}$\\
            $\X \gets \{\pos[s]\}$\\
            $\team$ recruits up to $4\ell-1$ new robots using $\textsc{DFSampling}(\squarereg,\X)$\\
            Move $\team$ to the center of $\squarereg$
            \item \label{item:round} \textbf{Round $k \geq 1$} for a team $\team$ in square $\squarereg$:
            \begin{enumerate}[label=\textbf{(\roman*)}]
                \item \textsf{\textbf{Termination}}\\ 
                If $|\team| < 4\ell$: do a centralized awakening of sleeping robots in $\squarereg$ and terminate.
                \item \textsf{\textbf{Partition}} \\ 
                Partition $\squarereg$ (resp. $\team$) into $4$ sub-squares $\squarereg[1], \squarereg[2], \squarereg[3], \squarereg[4]$ (resp. $4$ teams $\team_1, \team_2, \team_3, \team_4$ of $\ell$ robots each). \\
                Each team $\team_i$ performs in parallel:
                \begin{enumerate}[label=\textbf{(\roman*)}]
                    \setcounter{enumiii}{2}
                    \item \textsf{\textbf{Exploration}} \\
                    Using 4 times routine \explore, collectively explore $\sep(\squarereg[i])$.\\
                    Save in $\X_i$ the set of positions of $\sep(\squarereg[i])$ where a robot was found asleep, and those where a robot has already been awakened.
        
                    \item \textsf{\textbf{Recruitment}} \\
                    $\team_i$ recruits $\ell'_i$ robots  by using $\textsc{DFSampling}(\squarereg[i], \X_i)$\\
                    Move $\team_i$ to the center of \squarereg.
                \end{enumerate}
                
                \setcounter{enumii}{4}
                \item \textsf{\textbf{Reorganization}} \\
                Wait until the four teams $\team_i$ can merge and share their variables.\\
                Reorganize robots in teams $\team'_i$ of size $\ell'_i$ by square of origin $\squarereg[i]$.\\
                For each team $\team'_i$, do in parallel: 
                \begin{itemize}
                    \item Move $\team'_i$ to the center of $\squarereg[i]$ 
                    \item Go to Round $k+1$ with team $\team'_i$ and square $\squarereg[i]$.
                \end{itemize}
            \end{enumerate}
        \end{enumerate}
    \end{tcolorbox}
    \caption{Description of \Aseparator}
    \label{fig:code-aseparator}
\end{figure}

\begin{theorem}
    \label{th:upper}
    \Aseparator solves the dFTP, given every admissible tuple $(\ell,\rho,n)$, with a makespan $O(\rho + \ell^2 \clog(\rho/\ell))$.
\end{theorem}

In Section~\ref{sec:asep-proof} we prove Theorem~\ref{th:upper} in two steps: Lemma~\ref{lem:Aseparate-valid} guarantees that $\Aseparator$ solves the dFTP, and Lemma~\ref{lem:Aseparate-makespan} its makespan. Informally, we have $O(\clog {(\rho/\ell)})$ rounds (Lemma \ref{lem:nbRounds}) and Round $k\geq 1$, takes $O(\ell^2+\rho/2^k)$ time units (Lemma \ref{lem:durationRound}).

We also provide a lower bound on the makespan of any algorithm solving the dFTP.
This is done by building an $n$-point set $\P$ depending on the considered algorithm.
Details of the proof are given in Section~\ref{sec:lwBound}
\begin{restatable}[Lower Bound without energy constraint]{theorem}{lowerBound}
    \label{th:lower}
    For every admissible tuple $(\ell,\rho,n)$ and algorithm $\mathcal{A}$ solving the d-FPT, there exists an $n$-point set $\P$ and a source $s$ such that $\ellstar \leq \ell$ and $\rhostar \leq \rho$ such that the makespan of the execution of $\mathcal{A}$ with source $s$ and inputs $(\ell,\rho,n)$ on $\P$ is $\Omega(\rho + \ell^2 \clog(\rho/\ell))$.
\end{restatable}

\section{With energy constraint}\label{sec:energy}
The construction used in the proof of Theorem~\ref{th:lower} can be adapted to obtain Theorem~\ref{th:impossibility}.
Details are given in Section~\ref{sec:impossibility}.
\begin{restatable}[Lower Bound on the energy budget]{theorem}{impossibility}
    \label{th:impossibility}
    For every admissible tuple $(\ell,\rho,n)$ and algorithm $\A$ with energy budget $B < \pi(\ell^2 - 1)/2$, there exists an $n$-point set $\P$ and a source $s$ such that $\ellstar \leq \ell$ and $\rhostar \leq \rho$ such that the execution of $\mathcal{A}$ on $\P$ does not wake-up any robot.
\end{restatable}

We use a Breadth First Search based strategy to define $\Agrid$, an algorithm that takes as input the parameter $\ell \geq \ellstar$. 
This algorithm is then combined with $\Aseparator$ to get $\Awave$, which provides a smaller makespan.

To sum up, \Agrid works as follows: the plane is partitioned into squares of width $2 \ell$ centered at positions $\{(2k\ell,2k'\ell) \mid (k,k') \in \mathbb{Z}^2\}$. We wake-up the square $\squarereg$ containing $s$ in time $t(\squarereg) = R^2 + (10 + \sqrt{2})R$ (Corollary~\ref{cor:simple-square-wup}).
Then every square containing a new awake robot try to wake up the $8$ adjacent squares of a square ordered in a counter-clockwise order.
A precise description of \Agrid is provided in Section~\ref{sec:A-Grid}.

\begin{restatable}{theorem}{Agridmakespan}
    \label{th:Agridmakespan}
    \Agrid solves the dFTP, given every admissible tuple $(\ell,\rho,n)$, with an energy budget of $O(\ell^2)$ and with a makespan of $O(\ell \cdot \ecc_\ell)$.
\end{restatable}

The next algorithm, \Awave, is an adaptation of \Agrid. Roughly speaking, there are two changes: (1) the squares are now of width $8 \ell^2 \log_2{\ell}$; and (2) instead of using a simple process to explore and to wake up a square, we use $\Aseparator$ starting from a team of $4\ell$ robots for Rounds $k \geq 1$ and from the source $s$ for Round $0$. A precise description of \Awave is given in Section~\ref{sec:awave}.

\begin{restatable}{theorem}{Awavemakespan}
    \label{th:Awavemakespan}
    \Awave solves the dFTP, given every admissible tuple $(\ell,\rho,n)$, with an energy budget of $O(\ell^2 \log{\ell})$ and a makespan of $O(\ecc_\ell + \ell^2 \clog \ecc_\ell/ \ell)$. 
\end{restatable}

We can find a lower bound of $\Omega(\ecc)$ for a given of range of values for $\ecc$ (See construction in Section \ref{sec:lower-nrj}):

\begin{restatable}[Lower Bound for energy constrained]{theorem}{lowernrj}
    \label{th:lower-nrj}
    For every admissible tuple $(\ell,\rho,n)$, for every $B > \ell$, and for every $\ecc \in [\rho,\min\{n\ell-\rho/3,\floor{\rho^2/(2(B+1)) + 1}{}\}]$, there exists an $n$-points set $\P$ and a source $s$ of connectivity threshold $\ell$, radius $\rho$, and $\ell$-eccentricity $\ecc_\ell = \ecc$ such that the makespan of any algorithm $\mathcal{A}$ solving the dFTP for $(\P,s)$ given $(\ell,\rho,n)$ and energy budget $B$, is $\Omega(\ecc + \ell^2 \log{(\ecc/\ell)})$. 
\end{restatable}

\section{Discussion}
Although the input of our algorithms is $(\ell,\rho,n)$, it turns out that only the knowledge of $\ell$ as an upper bound of $\ellstar$ is required. The number of robots to awake $n$ is never used.

Only \Aseparator requires an upper bound $\rho$ on $\rhostar$.
We can  easily get a constant approximation of $\rhostar$ as follows: (1) we build a team of $4\ell$ robots using \dfsampling in time $O(\ell^2 \log \ell)$ and (2) we run  explorations of the $\ell$-separators squares of increasing width $\ell \cdot 2^i$ for $i=1,2,\ldots,k$ until we have an empty separator. Then we take $\rho=\ell \cdot 2^k$ and we can prove that it is a $3$-approximation of $\rhostar$. The overcost is $O(\ell^2 \log \ell + \sum_i^k \ell 2^i)=O(\ell^2 \log \ell+\ell 2^{k+1})=O(\ell^2 \log \ell+\rho)$. Either step (1) does not provide $4 \ell$ robots meaning that $\P$ is covered and that $\rhostar$ is deduced from the discovered robots or (2) we have $\log \ell =O(\log (\rhostar/\ell))$. To conclude, the total overcost of computing a constant approximation of $\rhostar$ from $\ell$ is of same order than \Aseparator.

We presented two algorithms that guarantee an optimal makespan under the hypothesis of knowing an upper bound on $\ellstar$.
\Aseparator provides an optimal makespan in $\Theta(\rhostar+\ell^2\log(\rhostar/\ell))$ without energy constraints, and \Awave provides an optimal makespan of $\Theta(\ecc_\ell+\ell^2 \log\ecc_\ell/\ell)$ with an energy budget in $\Theta(\ell^2 \log \ell)$.
Finally, $\Agrid$ provides a less efficient makespan (apart under some specific relations between $\ell,\rho$ and $\ecc_\ell$), but requires an optimal energy budget of $\Theta(\ell^2)$.

Some questions remain open: can we get an optimal algorithm with only $\Theta(\ell^2)$ energy?

Moreover, our lower bound with energy constraints holds for a specific range of $\ecc$, informally between $\rho$ and $c \cdot (\rho/\ell)^2$ for some constant $c$, whereas we know that $\ecc$ can approach $c \cdot \rho^2/\ell$: can we obtain a lower bound for a broader range of value of $\ecc$?

In our algorithms, robots has to communicate using rendez-vous. Is it possible to achieve efficient solution without face-to-face communication?

\section{Building blocks - details}

\paragraph{Main notations.}  Given $r \ge 0$ and $p\in \mathbb R^2$, we denote by $B_p(r)$ the disk of center $p$ and radius $r$. Given a real $\delta\ge 0$, and $\X \subset \mathbb{R}^2$, the \emph{$\delta$-disk graph} of $\X$ is the edge-weighted geometric graph whose vertex set is $\X$, two points $u,v \in \X$ being connected by an edge if and only if $u$ and $v$ are at (Euclidean) distance at most~$\delta$, and the weight of the edge corresponds to the distance between their endpoints.

We first emphasize a relationship between the different parameters of the point sets:

\begin{proposition}\label{prop:admissible}
For every point set $\P$ and source $s$, for any $\ell \geq \ellstar$:
\[
    0 < \ellstar \le \rhostar \le \ecc_\ell \le n \cdot\ellstar ~.
\]
\end{proposition}
Thus the input for our algorithm is a tuple $(\ell,\rho,n)$ such that $\ell \leq \rho \leq n\rho$.

\begin{proof}
    The first three inequalities are straightforward using definition. Note $\ell=\ellstar$ and $\rho=\rhostar$. The last one is because: (1) each edge of the $\ell$-disk graph of $\P\cup\set{s}$ has length at most $\ell$; and (2) a path from $s$ to any point of $\P$ in every tree spanning $n=|P|+1$ points contains at most $|P|$ edges. 
\end{proof}

\subsection{Explore}
\label{sec:explore-details}
We now detail how a team of $k$ robots can explore efficiently a square.

\rstexplore*

\begin{proof}
    Let us start by describing a \emph{Single exploration}, that is the path traveled by a single robot $r$ to explore a rectangular region of width $w$ and height $h$.
    Starting from a corner, keep on zigzaging rows by rows separated by $\sqrt{2}$ and every move of length $\sqrt{2}$, do a snapshot to discover up to radius $1$, as shown in Figure~\ref{fig:single-exploration}.
    Note that the algorithm may require an initial move and a final move, to go from $p$ to the origin of the path, and to go from the endpoint of the path to $p'$, both moves being bounded by $w+h$.
    The sum of the length of vertical paths is at most $h'$, and the length of each horizontal path is at most $w$.
    Furthermore, the number of horizontal path is $\lceil h/\sqrt 2 \rceil$, which means that the sum of the length of the horizontal paths is at most $w \lceil h/\sqrt 2 \rceil \in O(w\times h + w)$.
    By summing all three parts, we end up with a complexity of $\sf explore\_single$ in $O(w\times h + w + h)$.
    The validity of the procedure $\sf explore\_single$ comes with the fact that any disk of radius $1$ contains the square with identical center and width $\sqrt 2$.
    
    For a collaborative exploration with $k$ robots, the procedure separates the targeted rectangle into $k$ rectangles, each one being explored by a single robots as shown in Figure~\ref{fig:collaborative-exploration}. 
    Sub-rectangles are of equal dimensions $w' = w$ and $h' = \frac{h}{k}$
    The analysis conducted for the single exploration leads to a complexity of $O(\frac{w\times h}k + w + h)$ for this procedure.
    Since all robots can compute an upper bound on the time complexity for all robots, they can all reach the endpoint $p'$ at the same moment $t'$ and share the information they gathered in their own rectangles.
\end{proof}

\begin{figure}[h!]
    \centering
    \def\ray{100}
    \def\height{55}
    \def\length{187}
    \def\center{10}
    \def\halfstep{20}
    \def\step{40}
    \def\sechalfstep{60}
    \definecolor{col1}{gray}{0.6}
    \definecolor{col2}{gray}{0.9}
    \begin{subfigure}{0.45\textwidth}
        \begin{tikzpicture}[scale=0.8]
            \begin{scope}[shift=({-\ray pt,100pt})]
                \node (robot) at (-\halfstep pt,\height/2 pt) {$r$};
                \foreach \ycoord in {0, \step,...,\height} {
                    \foreach \xcoord in {0,\step,...,\length} {
                        \fill [color=col1] (\xcoord pt,\ycoord pt) rectangle (\xcoord+\halfstep pt,\ycoord+\halfstep pt);
                    }
                }
                \foreach \ycoord in {\halfstep} {
                    \foreach \xcoord in {\halfstep,\sechalfstep,...,\length} {
                        \fill [color=col1] (\xcoord pt,\ycoord pt) rectangle (\xcoord+\halfstep pt,\ycoord+\halfstep pt);
                    }
                }
                \foreach \ycoord in {0,\step,...,\height} {
                    \foreach \xcoord in {\halfstep,\sechalfstep,...,\length} {
                        \fill [color=col2] (\xcoord pt,\ycoord pt) rectangle (\xcoord+\halfstep pt,\ycoord+\halfstep pt);
                    }
                }
                \foreach \ycoord in {\halfstep} {
                    \foreach \xcoord in {0,\step,...,\length} {
                        \fill [color=col2] (\xcoord pt,\ycoord pt) rectangle (\xcoord+\halfstep pt,\ycoord+\halfstep pt);
                    }
                }
                \foreach \ycoord in {0, \step,...,\height} {
                    \foreach \xcoord in {0,\step,...,\length} {
                        \fill (\xcoord + \center pt, \ycoord + \center pt) circle[radius=1pt];
                        \draw (\xcoord + \center pt, \ycoord + \center pt) circle [radius=\center*1.414 pt];
                    }
                }
                \foreach \ycoord in {\halfstep} {
                    \foreach \xcoord in {\halfstep,\sechalfstep,...,\length} {
                        \fill (\xcoord + \center pt, \ycoord + \center pt) circle[radius=1pt];
                        \draw (\xcoord + \center pt, \ycoord + \center pt) circle [radius=\center*1.414 pt];
                    }
                }
                \foreach \ycoord in {0,\step,...,\height} {
                    \foreach \xcoord in {\halfstep,\sechalfstep,...,\length} {
                        \fill (\xcoord + \center pt, \ycoord + \center pt) circle[radius=1pt];
                        \draw (\xcoord + \center pt, \ycoord + \center pt) circle [radius=\center*1.414 pt];
                    }
                }
                \foreach \ycoord in {\halfstep} {
                    \foreach \xcoord in {0,\step,...,\length} {
                    .    \fill (\xcoord + \center pt, \ycoord + \center pt) circle[radius=1pt];
                        \draw (\xcoord + \center pt, \ycoord + \center pt) circle [radius=\center*1.414 pt];
                    }
                }
                \foreach \xcoord in {\halfstep, \step,...,\length} {
                    \draw[->, >=latex, color=red] (\xcoord-\center pt, \center pt) -- (\xcoord+\center pt, \center pt) ;
                    \draw[<-, >=latex, color=red] (\xcoord-\center pt, \halfstep+\center pt) -- (\xcoord+\center pt, \halfstep+\center pt) ;
                    \draw[->, >=latex, color=red] (\xcoord-\center pt, \step+\center pt) -- (\xcoord+\center pt, \step+\center pt) ;
                }
                \draw[->, >=latex, color=blue, very thick] (\center+9*\halfstep pt, \center pt) -- (\center+9*\halfstep pt, \halfstep+\center pt);
                \draw[->, >=latex, color=blue, very thick] (\center pt, \halfstep+\center pt) -- (\center pt, \step+\center pt);
            
                \draw (0,0) rectangle (\length pt, \height pt);
                \draw[<->, >=latex] (0, \height+25 pt) -- (\length pt, \height+25 pt) node[above, midway] {$w$} ;
                \draw[<->, >=latex] (\length+30 pt, 0) -- (\length+30 pt, \height pt) node[right, midway] {$h$} ;
            \end{scope}
        
            \begin{scope}[shift=({0,-50pt})]
                \fill[color=col1] (-100pt,0) rectangle (0,100pt) ;
                \fill[color=col2] (0,0) rectangle (100pt,100pt) ;
                \draw (-50pt,50pt) circle[radius=70.7pt] ;
                \draw (50pt,50pt) circle[radius=70.7pt] ;
                \fill (-50pt,50pt) circle[radius=2pt];
                \fill (50pt,50pt) circle[radius=2pt];
                \draw[->, >=latex, color=red] (-50pt,50pt) -- (50pt,50pt) node[below, midway] {$\sqrt 2$};
                \draw[->, >=latex] (-50pt,50pt) -- (-100pt,100pt) node[above, midway] {$1$};
            \end{scope}
        \end{tikzpicture}
        \caption{Scheme of the Single exploration procedure presented in the proof of Lemma~\ref{lem:collaborative-exploration}.\label{fig:single-exploration}}  
    \end{subfigure}
    \hfill
    \begin{subfigure}{0.45\textwidth}
        \centering
        \begin{tikzpicture}[scale=0.8]
            \draw (-\ray pt,-\ray pt) rectangle (\ray pt, \ray pt);
            \draw (0,0) circle (\ray pt);
            \foreach \width in {-100, -80,...,99} {
                \draw (-100 pt, \width pt) rectangle (100pt, \width+20 pt);
            }
            \node (r1) at (-110pt, -90pt) {$r_1$} ;
            \node (r2) at (-110pt, -70pt) {$r_2$} ;
            \node (ri) at (-110pt, 10pt) {$r_i$} ;
            \node (rk) at (-110pt, 90pt) {$r_k$} ;
            \draw[<->, >=latex] (110pt, -100pt) -- (110pt, 100pt) node[right, midway] {$h$} ;
            \draw[<->, >=latex] (120pt, -100pt) -- (120pt, -80pt) node[right, midway] {$h/k$} ;
            \draw[<->, >=latex] (-100pt, 115pt) -- (100pt, 115pt) node[above, midway] {$w$} ;
        \end{tikzpicture}
        \caption{Scheme of the construction of Lemma~\ref{lem:collaborative-exploration} applied to the collaborative exploration of a disk.\label{fig:collaborative-exploration}}
    \end{subfigure}
    \caption{Depiction of the exploration process.}
    \label{fig:exploration}
\end{figure}
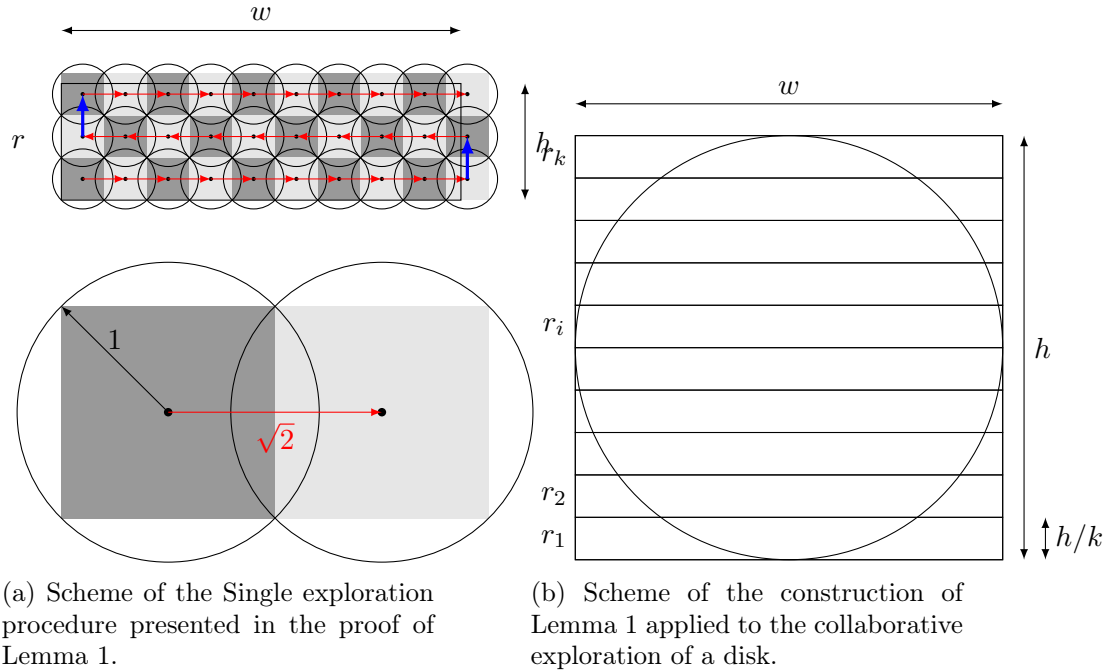

\subsection{Realization of a wake-up tree.}
\label{sec:wut-realization}

In this section we explain how to apply some results about Centralized Freeze Tag Problem to awake robots in a distributed setting.

Once an awake robot $\rob_0$ has learned the positions of sleeping robots in region of the plane of radius $R$, it can compute an arbitrary wake-up tree $\rob_0$ of weighted depth $O(R)$ \cite{BCGH24}. This step is called the \emph{centralized awakening} of the region by $\rob_0$. 
It remains to actually wake up robots. $\rob_0$ has to awake its first child in the wake-up tree and transmit the rest of the tree so the newly awake robot can help. More generally, each new awake robot participate to the propagation of the wake-up tree. 
The procedure starts by having $\rob_0$ moving to its child $\rob_1$ in $\wut_{r}$ (Algorithm~\ref{algo:propagate_wut}, line \ref{algo-line:wut-move}). 
When $\rob_0$ and $\rob_1$ are co-located, $\rob_0$ wakes up $\rob_1$ (line~\ref{algo-line:wut-wake-up}) and shares the wake-up tree $\wut$ (line~\ref{algo-line:wut-share}).
When a robot $\rob_i$ wakes up a robot $\rob_j$, both use a simple procedure to update in parallel what remains of the wake-up tree:
\begin{itemize}
    \item If the subtree of $\wut_i$ rooted in $\rob_j$ is empty, ie $\rob_j$ is a leaf, none of $\rob_i$ or $\rob_j$ has anything left to do (line~\ref{algo-line:wut-init-empty} and \ref{algo-line:wut-empty}).
    \item If $\rob_j$ has a unique child in $\wut_i$, $\rob_j$ updates its local wake-up tree  $\wut_j \gets \wut_i\setminus \{\rob_i\}$ so $\wut_j$ is rooted in $\rob_j$ (line~\ref{algo-line:wut-init-single}) and $\rob_i$ is done (line~\ref{algo-line:wut-single}). 
    $\rob_j$ wakes up its unique child in $\wut_j$ and recursively propagate $\wut_j$.
    \item If $\rob_j$ has two children, $\wut_i$ is separated into two sub-trees: $\rob_i$ keeps the right-hand sub-tree (line~\ref{algo-line:wut-update}) while $\rob_j$ keeps the left-hand sub-tree (line~\ref{algo-line:wut-init-update}). Both robots update their memory accordingly and moves to the unique child of their respective wake-up trees to continue the propagation.
\end{itemize}

\begin{algorithm}[H] \label{algo:propagate_wut}
\caption{Propagate Wake-Up Tree - Code for a robot $\rob$}
\Input{\wut}
    \lIf{$|\wut| = 0$}{\stop}\nllabel{algo-line:wut-init-empty}
    \lElseIf{$|\wut| = 1$}{$\wut \gets \wutchild (\wut)$}\nllabel{algo-line:wut-init-single}
    \lElseIf{$|\wut| = 2$}{$\wut \gets \wutchild_1 (\wut)$}\nllabel{algo-line:wut-init-update}
    
    \While{\wut}{
        $\dest \gets \root(\wut)$\;
        $\move(\dest)$\; \nllabel{algo-line:wut-move}
        $\wakeup(\dest, \wut)$\; \nllabel{algo-line:wut-wake-up}
        $\exchange(\wut)$\; \nllabel{algo-line:wut-share}
        \lIf{$|\wut| = 0$}{\stop}\nllabel{algo-line:wut-empty}
        \lElseIf{$|\wut| = 1$}{\stop}\nllabel{algo-line:wut-single}
        \lElseIf{$|\wut| = 2$}{
            $\wut \gets \wutchild_2 (\wut)$%
            \nllabel{algo-line:wut-update}
        }
    }
\end{algorithm}

\centralized*
\begin{proof}
    In \cite{BCGH24}, authors show that if a set of sleeping robot is contained within a disk of radius $R'$ around an awake robot, one can compute a wake-up tree  with makespan $5\sqrt{2}R'$.
    The proof follow by considering the disk of center $\pos[\rob]$ and radius $\frac{R}{\sqrt{2}}$. 
    \squarereg is entirely contains in that disk.
    By assumption $\rob$ known the positions of sleeping robots $\Sleeping$ within \squarereg. 
    Therefore $\rob$ can awake \squarereg by computing an propagating a wake-up tree with makespan $\frac{R}{\sqrt{2}} \sqrt{2} = 5R$.
    The correctness is due to the assumption that no robots other than $r \cup \Sleeping$ awake robots in \squarereg, which avoid any risk of conflicts.
\end{proof}

\subsection{Local Synchronization}
\label{sec:local-sync}

Let us prove Corollary~\ref{cor:simple-square-wup} whenever several awake robots located in the same square aim at waking up sleeping robots of the square. We assume that these awake robots start their action at the same time:
\simplesquarewup*
\begin{proof}
    Every awake robots move to the lower-left corner of $\squarereg$ in at most $\sqrt{2}R$. 
    One of them explores $\squarereg$ in time $R^2 + 5R$ (Lemma~\ref{lem:collaborative-exploration}) and move to the center of \squarereg.
    Among all awake robots of \squarereg, only that one does take action.
    Lemma~\ref{lem:centralized-awakening} apply: it can awake every sleeping robots in time $5R$.
\end{proof}

\subsection{\texorpdfstring{$\ell$}{l}-sampling and Eccentricity}

As discussed in Section~\ref{sec:dfsampling}, computing an $\ell$-sampling of a set of positions $\P$ is an efficient way to discover $\P$, and thus to wake up robots.
In this section we show some relationships between the values $\ell, \rho$ and $\ecc_\ell$.

\rstsamplingsize*
\begin{proof}
    For any $p_i,p_j \in \P'$, $B_{p_i}(\ell/2) \cap B_{p_j}(\ell/2) = \emptyset  $. 
    The area of each ball is $\pi \ell^2/4$. Although $p_i$'s are located within $\squarereg$, it may happen that only a fraction of $B_{p_i}(\ell/2)$ is contained within $\squarereg$. 
    This fraction is at least $1/4$ whenever $p_i$ is located at a corner of $\squarereg$. 
    Since $\bigcup_{i=1}^{|\P'|} (B_{p_i} \cap \squarereg) \subset \squarereg$, we have $|\P'| \pi \ell^2/16 \leq R^2$. 
\end{proof}

\begin{lemma}
    \label{lem:eccentricity}
    Let us consider a point set $(\P,s)$.
    For any $\ell \geq \ellstar$, we have $\ecc_\ell \in [\rhostar, \frac {12\rhostar^2}\ell]$, and for any position $p\in\P$, there exists a path from $s$ to $p$ in the $\ell$-disk graph which is at most $1+\frac{2\ecc_\ell}\ell$-hops long.
\end{lemma}
\begin{proof}
    The lower bound $\ecc_\ell \geq \rhostar$ is straightforward.

    Let $s,p_1,\ldots,p_k=p$ be a shortest path from $s$ to $p$ in the $\ell$-disk graph of $(\P,s)$, which is minimal in terms of hops.
    By definition, $\forall 1\leq i \leq k-2, |p_ip_{i+2}| > \ell$, otherwise $p_i$ and $p_{i+2}$ are connected in the $\ell$-disk graph, and therefore there exists another shortest path, smaller by one hop.
    Therefore, $d_\ell(s,p)$ is at least $\ell$ times the number of hops from $p_{2i}$ to $p_{2i+2}$: $d_\ell(s,p) \geq \ell \floor{k/2} \geq \ell \frac{k-1}2$.
    Since $d_\ell(s,p) \leq \ecc_\ell$ we conclude $k\leq 1 + 2\ecc_\ell/\ell$

    Note also that we can apply Lemma~\ref{lem:disk-covering} to the set of robots of even index in that path.
    Therefore we obtain $\floor{k/2} \leq \frac{16\rho^2}{\pi\ell^2}$ and thus $k \leq 1+\frac{32\rho^2}{\pi\ell^2} < \frac{36 \rho^2}{\pi \ell^2} < 12 \rho^2/\ell^2$.
    Furthermore, we have $d_\ell(s,p) \leq k\ell$, which imply that $d_\ell(s,p) \leq 12 \rhostar^2/\ell$.
    Since this is true for any $p$, and $\ecc_\ell = \max_{p\in \P} d_\ell(s,p)$ we deduce $\ecc_\ell(G) \leq 12 \rhostar^2 / \ell$.
\end{proof}

\subsection{Description and proof of Distributed \texorpdfstring{$\ell$}{l}-sampling}
\label{sec:detail-dfs}

We describe more formally the algorithm \dfsampling.
In a region \squarereg a team \team compute an $\ell$-sampling of \squarereg, starting from a given set of seeds $\X$. If robots have been awakened previously in \squarereg, we denote by $A$ the set of their initial positions. $A$ is assumed known by \team. The output is a set of points $\P'$ that is a $\ell$-sampling of \squarereg. In practice, we use \dfsampling to recruit a team $\team'$ of at most $4\ell$ robots identified by their positions $\P'$.

Let us start by defining a function $\textsc{Sort}(\X)$ to order positions of seeds $\X = \{X_i\}$ whenever $\X > 1$.
Seeds are ordered with respect to a projection on the borders of \squarereg in the clockwise order around the center of square $\squarereg$. 
More precisely, a seed $X_i$ is projected to the closest point of the border of $\squarereg$ breaking tie by choosing the first projected point in the clockwise order.

\begin{enumerate}
    \item $\P'=\emptyset$. Positions $\X=[X_1,X_2,\ldots,X_j]$ are ordered with $\textsc{Sort}(\X)$. Set $i=1$. Add $X_i$ to $\P'$.
    
    \item \label{item:dfs-step-2} Team $\team$ moves to $X_i$ and performs a Depth First Search traversal of the $2\ell$-disk graph $G$ induced by $\P \cap \squarereg$. 
    A neighbor $p'$ of a current position $p$ is either known if $p' \in A$ or is found using $\explore(B_p(2\ell),p)$. 
    The team explores as far as possible along each branch before backtracking and stop if $|\P'|=4\ell$. 
    A stack keeps track of the positions of $P'$ discovered so far along a specified branch which helps in backtracking of the graph. 
    We have one condition to move to $p'$: 
    \emph{a position $p'$ is added to $\P'$ if its distance to any point of $\P'$ is greater than $\ell$.}
    If the robot located at $p'$ is sleeping, it is awakened by the team $\team$ and added to $\team$.
    
    \item If $|\P'|<4\ell$ and $i<j$ then $i \gets$ the index of the next seed of $X_{i'}$ such that $B_{X_{i'}}(\ell)$ is not covered and repeat Step~\ref{item:dfs-step-2};
    \item return $\P'$
\end{enumerate}

We are now ready to prove the corresponding lemma:
\rstdfsampling*

\begin{proof}
    By construction, any $p'$ added to $\P'$ is at distance greater than $\ell$ from any other point of $P'$. 
    Inductively, any pair of positions in $\P'$ has a pairwise distance greater than $\ell$. Thus the output of \dfsampling is an $\ell$-sampling of $\squarereg$.
    Let us now analyze the time required to obtain that $\ell$-sampling.

    The Depth First Search traversal uses three operations: find neighbors, move to a neighbor and backtrack if no new neighbor is added.

    In the following $\team$ stands for the current team of robots.
    As soon as the robot team wakes up a sleeping robot, it is added to \team so the team size can increase.
    Let us start with seed $X_1$.
    The time to find the list of potential neighbors in the $2\ell$-disk graph is the time to explore $B_p(2\ell)$.
    From Lemma \ref{lem:collaborative-exploration}, it takes a time in $O(\ell^2/|\team|)$.
    Thus exploring and moving to a neighbor takes $O(\ell^2/|\team|+\ell)$.

    Let us first consider the case where the initial team $|\team|=1$ and $\X=A={p_s}$.
    Assume that at some point, $k$ robots were awakened and added to $\team$.
    This happens as soon as there is a sleeping robot on a position $p'$ that is added to $\P'$. 
    The time to find and wake up these $k$ robots is upper bounded by $\sum_{i=1}^k O(\ell^2/i) \in O(\ell^2 \clog k)$. 

    Let us focus on the time dedicated to moves. For recruiting $k \geq 2$ robots, $\team$ either have to move forward to a neighbor at a distance at most $2 \ell$, or backtrack to branch from a neighbor of a previously added robot.
    The amount of moves required to backtrack is at most $2 k \ell$ units since the tree corresponding to the Depth First Search traversal has $k$ edges of length at most $2 \ell$.
    In total, the total time of backtracking is less than $2k\ell \leq 8\ell^2 \in O(\ell^2)$ since $k \leq 4\ell$. 
    In total, it takes $O(\ell^2 \clog k)$ for this first case.

    The second case is whenever we start with a team of size $\ell$ and a set of seeds $\X = \P \cap \sep(\squarereg)$.
    The exploration is now done collaboratively with $\ell$ robots. By Lemma~\ref{lem:collaborative-exploration}, for any position $p' \in \P'$ the exploration phase takes $O(\ell)$ .
    As in the first case, the different moves take $O(k\ell)$.
    Thus the total time of exploration for the recruitment of $k$ robots takes $O(k \ell)$. 

    The second change is that the Depth First Search traversal starting from seed $X_1$ can stop before reaching $4 \ell$ robots.
    In this case the search goes on starting from the next seed.
    We have to add the time corresponding to the moves from one seed to another.
    Note that \team starts a branch from at most $4 \ell$ positions $X_{i_1},\ldots,X_{i_{4 \ell}}$ of $\X$.
    Assume that \team recruit $k$ seeds among these $4 \ell$ positions. 
    We rely on the fact that $\X \subseteq \sep(\squarereg)$ and the ordering computed by $\textsc{Sort}(\X)$ to bound the total duration of theses intermediates moves.
    
    Let $Y_{i}$ be the projected point of $X_i$ on the border of $\squarereg$. 
    By definition of the ordering $\textsc{Sort}(\X)$, the distance between two consecutive seeds $X_{i_j}$ and $X_{i_{j+1}}$ is at most $|X_{i_j}Y_{i_j}|+ |Y_{i_j},Y_{i_{j+1}}|+|X_{i_{j+1}}Y_{i_{j+1}}| \leq 2 \ell + d_1(Y_{i_j},Y_{i_{j+1}}) $ where $d_1(p,p')$ is the distance in the $L_1$-norm for $p,p' \in \mathbb{R}^2$. 
    Thus the length of $\sum_{j=1}^{k-1} |X_{i_j}X_{i_{j+1}}| \leq \sum_{j=1}^{k-1} d_1(Y_{i_j},Y_{i_{j+1}}) + 2 \ell \leq 4R + 2 k \ell$ since $\sum_{j=1}^{k-1} d_1(Y_{i_j},Y_{i_{j+1}})$ is at most the perimeter of the square. 
    To conclude, it takes $O(R + k \ell)$.
    
    We now focus on the properties (1) and (2).
     Assume that \dfsampling ends. 
     We show by contradiction that if $|\P'| < 4 \ell$ then $\squarereg$ is covered by $\P'$: $\forall p \in \P \cap \squarereg, \exists p' \in \P'$ such that $p \in B_{p'}(2\ell)$.

    Assume that there exists a point $p_0 \in \P \cap \squarereg$ that is not covered by $\P'$. 
    Since the connectivity threshold is less than $\ell$, the $\ell$-disk graph $G$ of $\P$ is connected and any pair of points in $\P$ are linked by a path in $G$. 
    In particular, there must exist a seed $p_x \in \X$ such that there is a path $b = (p_0, p_1, ..., p_x)$ in $G$ entirely contains in \squarereg: $\bigcup_{i=1}^x, p_i \subset \squarereg$. 
    It is trivial if $\P$ is fully contained in $\squarereg$ \ie when $R \geqslant 2\rhostar$.
    
    Else, if there is in $\P$ at least one point $q$ outside of $\squarereg$ and $\X =  \P \cap \sep(\squarereg)$. 
    From Lemma~\ref{lem:separator-path} any path from $p_0 \in \squarereg^{in}$ to $ q \in \squarereg^{out}$ contains at least one position $p_x \in \sep(\squarereg) \cap \X$.
    Consider a path $b$ from $p_0$ to $p_x$. 
    By assumption, $p_0$ is not covered by $\P'$. 
    Let $p_i$ be the first position on the path such that $p_i$ is not covered by $\P'$ and $p_{i+1}$ is covered. 
    By construction of $\P'$ if $p_{i+1}$ is covered there is a position $p'_{i+1} \in \P'$ contains the close neighborhood of $p_{i+1}$ in $G$ \ie $|p_{i+1}p'_{i+1}| \leqslant \ell$ or $p_{i+1} = p'_{i+1}$. 
    This implies that $B_{p'_{i+1}}(2\ell)$ covered by $P'$. 
    Since $|p_ip_{i+1}| \leqslant \ell$ and $|p_{i+1}p'_{i+1}| \leqslant \ell$ we get $|p_ip'_{i+1}| \leqslant 2\ell$ which conflicts with the assumption that $p_i$ is not covered.
\end{proof}

\section{Proofs of \texorpdfstring{\Aseparator}{Aseparator}}
\label{sec:asep-proof}
We first prove the following lemma, which will be useful in the proofs of the makespan of \Aseparator and \Awave.
\begin{lemma}
    \label{lem:BigO}
    For any $r\geq \ell \geq 1$ we have:
    $O(r + \ell^2 ( \clog \min\{\ell,r / \ell\} + \clog \frac{r}{\ell^{3/2}} ) ) \subseteq O(r + \ell^2 \clog (r / \ell) )$
\end{lemma}
\begin{proof}
    $r/\ell^{3/2}\leq r/\ell$ and we have two cases:
    \begin{enumerate}
        \item $\ell \geq r/\ell \geq r/\ell^{3/2}$ implying $\clog (\min\{\ell,r / \ell\} + \clog \frac{r}{\ell^{3/2}}) \leq 2 \clog r/\ell.$
        \item $1 \leq \ell \leq r/\ell \leq r/\ell^{1/2}$.
    \end{enumerate}
    In this second case, $\clog{\ell}+\clog{r/\ell}=\log(1+\ell)+\log(1+r/\ell^{3/2}) \leq \log (1+\ell+r/\ell^{3/2}+r/\ell^{1/2}) \leq \log (4r/\ell^{1/2})$.
    We also have $\log( r/\ell^{1/2})=\log\ell^{1/2} + \log r/\ell \leq \log\ell + \log r/\ell \leq 2 \log r/\ell$.
    Thus $\clog (\min\{\ell,r / \ell\} + \clog \frac{r}{\ell^{3/2}}) \leq \log 4 + 2\log(r/\ell) \in O(\clog (r/\ell))$.
\end{proof} 

\paragraph{Makespan \Aseparator}
\begin{restatable}[Makespan of \Aseparator]{lemma}{separatorMakespan}
    \label{lem:Aseparate-makespan}
    For every admissible tuple $(\ell, \rho, n)$ and for any $n$-point set $\P$ and source $s$ such that $\ellstar \leq \ell$ and $\rhostar \leq \rho$, the execution of $\Aseparator$ with source $s$ and inputs $(\ell, \rho, n)$ on $\P$ terminates in time $O(\rho + \ell^2 \clog(\rho/\ell))$
\end{restatable}
\begin{proof}
    Since \Aseparator is decomposed into rounds, the proof is done by bounding the duration of a Round $k$, denoted $t_k$, and the number of rounds until \Aseparator terminates.
    
    The duration of Round 0 is the duration of recruiting a team \team of size up to $4\ell$ within a square of width $2\rhostar$. 
    The recruitment is done by computing a $\ell$-sampling of the square with \dfsampling.
    From Lemma~\ref{lem:disk-covering} there is at most $32 \rhostar^2 / (\pi \ell^2)$ points in a $\ell$-sampling of a square of width $2\rhostar$, therefore $|\team| = \min(4\ell, 32 \rhostar^2 / (\pi \ell^2))$.
    Lemma~\ref{lem:dfsampling} guarantee that the time of \dfsampling depends of the size of the output team. 
    Hence the duration of Round 0 is:
    \begin{equation}
    \label{eq:Asep-rd0}
        t_0 \in O(\ell^2 \clog (\min\{\ell,\rhostar / \ell\})) \subseteq O(\ell^2 \clog (\min\{\ell,\rho / \ell\})).
    \end{equation}

    For $k \geq 1$, we denote by $\squarereg^{(k)}$ and $\team^{(k)}$ a square region and a team given as inputs at Round $k$.
    Let $R^{(k)}$ be the width of $\squarereg^{(k)}$. Note that for a given $k>1$, several squares are treated in parallel so neither $\squarereg^{(k)}$ nor $\team^{(k)}$ are unique.
    Given $k$, $R^{(k)} = \frac{2 \rho }{2^{k-1}}$ since at each round, the square $\squarereg^{(k)}$ is partitioned into four squares implying that $R^{(k + 1)}$ is the half of $R^{(k)}$, starting with $R^{(1)} = 2\rho$.
    
    Round $k \geq 1$ with input $\squarereg^{(k)}$ and $\team^{(k)}$ is a \emph{terminating} Round if it begins with $|\team^{(k)}| < 4\ell$, otherwise it is a \emph{partitioning} Round. 
    A terminating round only takes time $O(R^{(k)}) \subseteq O(\rho)$ to move toward the center of the sub-square and wake-up remaining sleeping robots discovered in the \dfsampling execution.
    
    Let us focus on partitioning rounds. 
    We show in Lemma~\ref{lem:durationRound} that for $k \geq 1, t_k = O(R^{(k)} + \ell^2)$.

    Let $k_{max}$ be the maximum number of Rounds. From Lemma~\ref{lem:nbRounds}, $k_{max} = O(\clog \frac{\rho}{\ell^{3/2}})$. The proof follows by summing the duration of Rounds from $k = 0$ to $k_{max}$
    \begin{align}
        T_{\Aseparator} &= \sum^{k_{max}}_{k=0} t_k = t_0 + \sum^{k_{max}}_{k=1} t_k  \\
        ~ & \in O(\ell^2 \clog (\min\{\ell,\rho / \ell\})) + \sum^{k_{max}}_{k=1}O(R^{(k)} + \ell^2)\\
        ~ & \in O(\ell^2 \clog (\min\{\ell,\rho / \ell\})) + \sum^{k_{max}}_{k=1} O(\frac{2\rho}{2^{k-1}}) + k_{max}O(\ell^2)\\
        ~ & \in O(\ell^2 \clog (\min\{\ell,\rho / \ell\})) + O(\rho) + O(\clog\frac\rho{\ell^{3/2}})O(\ell^2)\\
        ~ & \in O(\rho + \ell^2 ( \clog \min\{\ell,\rho / \ell\} + \clog \frac{\rho}{\ell^{3/2}} ) ) \\
        ~ & \in O(\rho + \ell^2 \clog (\rho / \ell) )
    \end{align}
    The last line comes from Lemma~\ref{lem:BigO}.
\end{proof}

\begin{lemma}
    \label{lem:Aseparate-valid}
    For every admissible tuple $(\ell, \rho, n)$ and for any $n$-point set $\P$ and source $s$ such that $\ellstar \leq \ell$ and $\rhostar \leq \rho$, the execution of $\Aseparator$ with source $s$ and inputs $(\ell, \rho, n)$ on $\P$ eventually wakes up all robots.
\end{lemma}
\begin{proof}
    Consider an execution of \Aseparator on $\P$. 
    From Lemma~\ref{lem:nbRounds} an execution terminates after a finite number of Rounds.
    We show that when the execution terminates, any robot $\rob \in \R$ initially located at $p \in \P$ is awakened. 
    Let $k$ be the largest $k$ such that at Round $k$ there is a input square $\squarereg^k$ containing $p$.
    
    Either $\rob$ is awakened and then $\rob \in \team^k$, or it is asleep.
    In this case, we show that \rob must be awakened at the end of Round $k$.
    
    Firstly, Round $k$ on $\squarereg^k$ and $\team^k$ must be a terminating Round.
    If it is not the case, \ie it is a partitioning Round, then for $\{i \in [1, 4] \mid \team'_i \neq \emptyset \}$ sub-square $\squarereg[i]^k$ is an input of Round $k+1$. 
    Since robots are assigned to teams $\team'_i$ based on their initial position, if $\team'_i$ is empty then $\P \cap \sep(\squarereg[i]^k) = \emptyset$, by Corollary~\ref{cor:separator}, $\squarereg[i]^k$ is empty and does not contain $p$. 
    Hence if Round $k$ is a partitioning Round then at Round $k+1$ there must be an input $\squarereg^{k+1}$ containing $p$, which breaks the assumption of $k$ being maximal.
    
    If Round $k$ is terminating, by definition of \Aseparator, the input team $\team^k$ contains less than $4 \ell$ robots. 
    It implies that during the recruitment phase of Round $k-1$, \dfsampling in $\squarereg^k$ has stopped before constituting a team of $4\ell$ robots. 
    Lemma~\ref{lem:dfsampling} guarantee that $\squarereg^k$ is covered by $\team^k$ \ie each position of $\P \cap \squarereg^k$ is known by $\team^k$ at the beginning of Round $k$.
    Therefore $\rob$ must be awakened when $\team^k$ operates a centralized awakening of sleeping robots of $\squarereg^k$.
\end{proof}

\paragraph{Duration of Rounds}
\begin{restatable}[Duration of Rounds]{lemma}{durationRound}
    \label{lem:durationRound}
    For $k \geq 1$, the duration of Round $k$ is in $O(R^{(k)} + \ell^2)$, $R^{(k)}$ being the width of sub-squares.
\end{restatable}

\begin{proof}
    We show that:
    \begin{itemize}
        \item a terminating Round has a duration $O(R^{(k)})$
        \item a partitioning Round has a duration $O(R^{(k)} + \ell^2)$
    \end{itemize}
    
    Let us begin with the termination case, that is an input $\squarereg^{(k)}$ and $|\team^{(k)}| < 4\ell$. The Round $k$ consists of using a centralized algorithm to wake up sleeping robots in $\squarereg^{(k)}$. 
    Recall that  $\team^{(k)}$ is recruited during the \recruit phase of Round $k-1$ by sampling $\squarereg^{(k)}$.
    Since $|\team^{(k)}| < 4\ell$ Lemma~\ref{lem:dfsampling} guarantees that $\squarereg^{(k)}$ is covered by the subset of $\team^{(k-1)}$ sampling $\squarereg^{(k)}$. 
    At the end of Round $k-1$ the knowledge of each subset of $\team^{(k-1)}$ is shared, thus sleeping robots of $\P \cap \squarereg^{(k-1)}$ are known. 
    The smallest disk containing $\squarereg^{(k)}$ is of radius $(1/\sqrt{2})R^{(k)}$. By Lemma~\ref{lem:centralized-awakening},  one robot of $\team^{(k)}$ can wake up any set of robots sleeping in this disk with a makespan at most $O(R^{(k)})$.
    
    Let us now consider a partitioning Round $k$ with a non-terminating input $\squarereg^{(k)}$ and $|\team^{(k)}| \geq 4\ell$. It consists of having four teams of $\ell$ robots exploring and recruiting within sub-squares of $\squarereg^{(k)}$ of width $R^{(k+1)}$. 
    The separator of a sub-square can be decomposed into $4$ rectangles of dimension $\ell \times R^{(k+1)}$. By Corollary~\ref{lem:collaborative-exploration} a team of $\ell$ robots can explore each rectangle in $O(\frac{\ell \times R^{(k+1)}}{\ell} + \ell + R^{(k+1)}) = O(R^{(k+1)} + \ell)$.
    The recruitment is done by \dfsampling from a set of positions $\X$ in $\sep(\squarereg^{(k)})$. 
    From Lemma~\ref{lem:dfsampling} the $\ell$-sampling of $\squarereg^{(k+1)}$ with a team of $\ell$ robot is done in time $O(R^{(k+1)}) + \ell^2)$. Recall that $R^{(k+1)} =\frac{1}{2}R^{(k)}$, we get that the duration of a partitioning Round $k$ is $O(R^{(k)} + \ell^2)$ 
\end{proof}

\paragraph{Number of Rounds}
\begin{restatable}[Number of Rounds]{lemma}{nbRounds}
    \label{lem:nbRounds}
    The number of rounds of \Aseparator is $1$ for $\rhostar \leq \frac{\ell^{3/2}}{8}$ and $\log (1+\sqrt{8/\pi} \rho / \ell^{3/2}) \in O(\log \frac{\rho}{\ell^{3/2}})$ otherwise.
    Both values are in $O(\clog \frac{\rho}{\ell^{3/2}})$.
\end{restatable}

\begin{proof}
We have two cases depending on $\rhostar$, being large or small with respect to $\ell$. Set $R=2 \rhostar$.

If $R \leq \frac{\ell^{3/2}}{4}$, from Lemma \ref{lem:disk-covering}, any $\ell$-sampling has size at most $\frac{16 R^2}{4\pi \ell^2} \leq 4 \ell /\pi < 4l$. Algorithm \Aseparator stops in Round $1$.

From now, we assume that $R > \frac{\ell^{3/2}}{4}$. Let us compute the smallest $k$ such that every execution of \dfsampling to $\squarereg^{(k)}$ of width $R^{(k)}$ outputs a $\ell$-sampling $\team'^{(k)}$ strictly smaller to $4 \ell$. In this case, \Aseparator ends at the beginning of Round $k+1$ during the termination phase.  

From Lemma~\ref{lem:disk-covering}, we have
\begin{align}
    |\team^{(k)}| & \leq \frac{16 (R^{(k)})^2}{\pi \ell^2} \\
    \text{Recall that } R^{(k)} &= \frac{2\rho}{2^{k-1}} \\
    \text{We get } R^{(k)} \leq \frac{\ell^{3/2}}{4} & = \frac{256 \rho^2}{2^{2k}} \pi \ell^2 <  4\ell\\
    \text{if } 2^{2k}  &> \frac{8 \rho^2}{\pi \ell^3} \\
    k &> \log (\sqrt{8/\pi} \rho / \ell^{3/2}).
\end{align}
\end{proof}

\section{Proofs of \texorpdfstring{\Agrid}{Agrid} and \texorpdfstring{\Awave}{Awave}}
\label{sec:A-Grid-A-Wave}
\subsection{Algorithm \texorpdfstring{\Agrid}{Agrid}}
\label{sec:A-Grid}
We now describe in detail Algorithm \Agrid:

Let $t(\ell)$ be an upper bound on the time required for the exploration and centralized awakening of a square region of width $2\ell$ by one robot.
By Corollary~\ref{cor:simple-square-wup}, $t(\ell) \in O(\ell^2)$.

\begin{itemize}
   \item \emph{Round $0$ - Initialization}: \\
     $\squarereg \leftarrow $ square of width $R = 2\ell$ and centered at the source robot $s$ \\
    Explore and wake-up \squarereg 
    \item \emph{Round $k \geq 1$ starting at time $t_k=t(\ell)+8(k-1)(t(\ell)+ \sqrt{2} R)$}: \\
    For every robot $\rob$ awakened in Round $k-1$ do in parallel: \\
    $\squarereg \gets $ the square containing $\rob$\\
    For $i \in [1..8]$ do
    \begin{enumerate}
        \item move to the lower-left corner of the $i$-th adjacent squares $\squarereg[i]$ of $\squarereg$
        \item Wait until time $t_k+ (t(\ell)+\sqrt{2} R)i$
        \item Explore and Wake-up \squarereg[i] 
    \end{enumerate} 
\end{itemize}

Let us remind the result on the makespan of \Agrid:
\Agridmakespan*

\begin{proof}
The duration of Round $0$ is $t(\ell) \in O(\ell^2)$. 
Then, in Round $k \geq 1$, we have to wake-up at most $8$ squares in time $8 t(\ell)$.
Since every square are adjacent, it takes at most $\sqrt{2}R$ to go the next adjacent square to wake-up.
In total, every round takes $O(\ell^2)$. Every robot  awakened in Round $k$ participates to Round $k+1$ but stops do at the end of Round $k+1$. Thus it only has to use $O(\ell^2)$ amount of energy.

Let us now upper bound the number of rounds. 
Given any robot $\rob$ located at position $p$, let us consider a shortest path from $s$ to $\rob$ in the $\ell$-disk graph of $(\P,s)$ minimizing the number of hops and the sequence of squares $S=s_1,\ldots,s_k$ crossed by this path. 
Since the path can cross several times a same square and every robot of a given square are awakened in the same round, the number of rounds required to wake up the sequence of squares is smaller or equal to the number of hops of the shortest path. 
From Lemma~\ref{lem:eccentricity}, we know that there exists a path of at most $2 \ecc_\ell/\ell$ hops implying the same upper bound for the number of rounds. 
Since every round takes a time $O(\ell^2)$, the makespan is $O(\ell \cdot \ecc_\ell)$. 
Furthermore since every awakened robot in Round $k$ only move in Round $k$ and $k+1$ and stop, robots only need $O(\ell^2)$ energy.  
\end{proof}

\subsection{Algorithm \texorpdfstring{\Awave}{Awave}}
\label{sec:awave}

We now describe in detail Algorithm \Awave:

Let $t(R)$ be an upper bound on the time required for \Aseparator to wake up all robots of a square of width $R$ starting from a team of size $4\ell$.
By Theorem~\ref{th:upper}, $t(R) \in \Theta(R + \ell^2\log\ell)$.
\begin{itemize}
   \item \emph{Round $0$ - Initialization}: \\
     $\ell \gets \max\{\ell, 4\}$\\
     $\squarereg \leftarrow $ square of width $R = 8\ell^2 \log_2 \ell$ and centered at the source robot $s$ \\
    Wake-up \squarereg using \Aseparator \\
    If there is no robot in $\sep(\squarereg)$, \Awave stops
    
    \item \emph{Round $k \geq 1$ starting at time $t_k=t(R)+8(k-1)(t(R)+ \sqrt{2} R)$}, for every $\rob$ awakened in Round $k-1$: \\
    $\squarereg(\rob) \gets$ square containing $\rob$ \\
    Move $\rob$ to the lower-left corner of $\squarereg(\rob)$ to build a team $\team_\rob$\\
    For every team $\team_\rob$ such that $|\team_\rob| \geq 4\ell$ do in parallel:\\
    For $i \in [1..8]$ do 
    \begin{enumerate}
        \item move to the lower-left corner of the $i$-th adjacent squares $\squarereg[i]$ of $\squarereg(\rob)$
        \item Wait until time $t_k+ (t(R)+ \sqrt{2} R)i$
        \item Wake-up \squarereg[i] using $\Aseparator$ within $\squarereg[i]$ with $\team_r$ in Round $k$
    \end{enumerate} 
\end{itemize}

Let us prove Theorem~\ref{th:Awavemakespan}
\Awavemakespan*
\begin{proof}
    We first prove that $\Aseparator$ solves the dFTP with energy budget $O(\ell^2 \log \ell)$ and a makespan of $O(\ecc_\ell + \ell^2\log\ell)$, and explain in a second time how we obtain the announced makespan.
    Firstly, note that $t(R) \in \Theta(\ell^2\log\ell)$.

    To begin, every round takes a time at most $O(R)=O(\ell^2 \log \ell)$. The only robots that participates in Round $k$ has been awakened only in Round $k+1$, it means that the energy budget required per robot is $O(\ell^2 \log \ell)$. 

    At Round $0$, $\Aseparator$ wakes up every robot of the initial square $\squarereg$. If there is a robot within $\sep(\squarereg)$, it means that $\squarereg$ contains at least $\floor{\frac{R/2-\ell}{\ell}} \geq \frac{R}{2\ell}-2 = 4 \ell \log_2 \ell -2 \geq 4 \ell$ robots since $\ell \geq 4$. 
    Thus, $\squarereg$ has enough robots to apply \Aseparator in every adjacent squares at Round $1$, and every robot within the $8$ adjacent squares are awakened during Round $1$.

    Let us now prove that all robots are awakened by the algorithm by providing an upper bound on the number of rounds to wake up a robot $\rob$ located outside the $9$ central squares. 
    Let $G$ be the $\ell$-disk graph of $\P$ and let $P=s,\rob_1,\rob_2\ldots,\rob$ be a shortest path in $G$ from $s$ to $\rob$, and let $s,\rob_1,\rob_2,\ldots,\rob_{j'}$ be the maximal subpath of $P$ of awake robots at the end of Round $k$.
    We now consider a robot $\rob_j$ with $j<j'$ such that $R-3\ell < d_\ell(\rob_j,\rob_{j'}) \leq R-2\ell$.
    Such a robot exists since by definition of $s$, $d_\ell(s,r) > R$.

    First, let us show that the subpath $P'=\rob_j,\rob_{j+1},\ldots,\rob_{j'}$ is contained within $\squarereg(\rob_{j'})$ and at most $2$ squares simultaneously adjacent to $\squarereg(\rob_{j'})$ and $\squarereg(\rob_{j'+1})$.
    Let $\squarereg$ be a square, adjacent to $\squarereg(\rob_{j'})$ but not adjacent to $\squarereg(\rob_{j'+1})$. 
    The Euclidean distance between $\rob_{j'}$ and $\squarereg$ is greater than $R-\ell$ since $\rob_{j'}$ is at distance at most $\ell$ from $\squarereg(\rob_{j'+1})$. 
    Thus $P'$ cannot cross $\squarereg$. 
    Moreover, it is impossible that $P'$ crosses two adjacent squares  of $\squarereg(\rob_{j'})$, $\squarereg$ and $\squarereg'$, such that $\squarereg$ and $\squarereg'$ are not adjacent because $P'$ has a length strictly smaller then $R$.
    This guarantees the announced property on $P'$.

    Secondly, let us show that $\squarereg(\rob_{j'+1})$ is awakened at Round $k+1$. 
    To have this guarantee, we need to have at least $4\ell$ awakened at Round $k$ in an adjacent square. 
    Since $P'$ is contained within at most $3$ adjacent squares of $\squarereg(\rob_{j'+1})$, the most populated adjacent square contains at least $\floor{\frac{R-2\ell}{\ell}}/3=\floor{\frac{R}{\ell}}/3-2/3 \geq (R/\ell-1)/3-2/3=R/(3 \ell)-1$ awake robots. 
    If $R=8 \ell^2 \log \ell$ and $\ell \geq 4$, $R/(3 \ell)-1 \geq 16 \ell/3 -1 \geq 4 \ell$. 

    Now, take the maximal subpath $P''$ of $P$ of length smaller than $R-2\ell$ but starting from $\rob_{j'+1}$. 
    We can show similarly as before that this path is either within $3$ adjacent squares awakened in the worst case in Rounds $k+1$, $k+2$ and $k+3$ or that this path ends in already awakened square. 
    In any case, within $3$ rounds, the length of the maximal subpath of $P$ of awake robots at the end of Round $k+3$ has increased of at least $R-3\ell > 5 \ell^2 \log \ell$ units. 
    Thus the number of rounds to wake up the robots of highest eccentricity takes $O(\ecc_\ell/\ell^2 \log \ell)$ rounds.
    The makespan is then bounded by $O(\ecc_\ell)$.
    
    Finally, let us be more precise on the bound on the makespan.
    We first consider the case where $\ecc_\ell \leq \ellstar ^{3/2}/16 \leq  \ell^{3/2}/16$.
    We immediately have $\ecc_\ell < R/2$ and therefore \Awave terminates at Round~0.
    But since $\ecc_\ell \leq \ell^{3/2}/16$, we have by Proposition~\ref{prop:admissible}, $\rhostar \leq \ell^{3/2}/16$ and so, by Lemma~\ref{lem:nbRounds}, $\Aseparator$ terminates at Round 0.
    In that case, the makespan of \Awave is therefore $t_0 \in O(\ell^2 \clog (\min\{\ell,\rhostar / \ell\}))$ as stated in Equation~\ref{eq:Asep-rd0}.
    Since by Proposition~\ref{prop:admissible} and Lemma~\ref{lem:eccentricity}, $\rhostar/\ell \leq \ecc_\ell/\ell \leq \rhostar^2/(12\ell^2)$, we obtain an overall complexity of $\Awave$ is in $O(\ell^2 \clog (\min\{\ell,\ecc_\ell / \ell\}))$.
    
    Otherwise, we have $\ecc_\ell \geq \ell^{3/2}/16$, which means that $\ecc_\ell/\ell \geq \sqrt\ell / 16$ and thus that $\min \{\sqrt \ell/16, \ecc_\ell/\ell\} = \sqrt\ell/16$. 
    Therefore, $O(\ell^2 \clog\ell) = O(\ell^2 \clog(\min \{\ell, \ecc_\ell/\ell\}))$.
    This guarantee an overall complexity in $O(\ecc_\ell + \ell^2\clog \min\{\ell, \ecc_\ell/\ell\})$ for \Awave.
    
    To summarize, we have an overall complexity in $O(\ecc_\ell + \ell^2 \clog (\min\{\ell,\ecc_\ell / \ell\})) \subseteq O(\ecc_\ell + \ell^2 \clog (\min\{\ell,\ecc_\ell / \ell\}) + \ecc_\ell/\ell^{3/2})$.
    By Lemma~\ref{lem:BigO}, we conclude that \Awave has a makespan in $O(\ecc_\ell + \ell^2 \clog \ecc_\ell/\ell)$.
\end{proof}

\section{Lower bounds - details}
\subsection{Without energy constraint}
\label{sec:lwBound}
Let us recall and prove Theorem~\ref{th:lower}:
\lowerBound*

Let us consider an algorithm $\A$ solving the dFTP, and an admissible tuple $(\ell,\rho,n)$.
To prove our lower bound, we first define some disjoint regions $D_{c}$ of $B_{0,0}(\rho)$ as pictured in Figure~\ref{fig:LB-construct}.
Each region has area $\Theta(\ell^2)$, and we prove in Lemma~\ref{lem:grid-included} that there are $\Theta(\rho^2/\ell^2)$ such regions.
The general idea is to place one sleeping robot in each region, depending on the behavior of $\A$, in a way that guarantees that awake robots have to visit the entire region before they find the sleeping robot in it.
We prove in Lemma~\ref{lem:spheres-connectivity} that the set of points we propose is $\ell$-connected, which makes the construction valid.

\paragraph{Centers and connectivity.}
We define $\centers = \{ (x,y) \in (\frac \ell2 \mathbb Z)^2 \mid \sqrt {x^2 + y^2}\leq \rho-\frac \ell4 \}$ the vertices of the $\frac \ell2$-grid,  restricted to the disk of center $(0,0)$ and radius $\rho-\frac \ell4$.
We say that two elements $(x,y)$ and $(x',y')$ in $\centers$ are \emph{adjacent} if $x'=x \wedge y'=y \pm \frac \ell2$ or $x'=x \pm \frac \ell2 \wedge y'=y$.
A subset $C\subseteq \centers$ is \emph{connected} if for any $c,c'\in C$, there exists a path of adjacent elements of $C$ with extremities $c$ and $c'$.

We also define $\centers^* = \centers\setminus\{(0,0)\}$ and $m = \min (n,  |\centers^*|)$.
Note that, since $(\ell,\rho,n)$ is admissible, and by Lemma~\ref{lem:grid-included}, we have $m \geq \rho/\ell$.
We are now going to prove a lower bound of $\Omega(\ell^2 \clog m)\subseteq \Omega(\ell^2 \clog \rho/\ell)$.

We denote by $\centers_m$ some subset of $\centers^*$ with $m$ elements such that $\centers_m \cup \{(0,0)\}$ is connected, and contains at least $\{ (0,\frac\ell2), (0,2\frac\ell2), \dots, (0,\floor{\rho/\ell}\frac\ell2) \}$.
Note that this is always feasible because, on the one hand, $m \geq \floor{\rho/\ell}$, and on the other hand, $\floor{\rho/\ell}\frac\ell2 \leq \rho/2 \leq \rho - \frac \ell4$, so all these points are actually in $\centers^*$.

\paragraph{Disks, $\ell$-connectivity.}
For any $c = (x,y) \in \centers_m$, let us define the disk of center $c$ by $D_{c} = B_{c}(\ell/4)$, of area $\frac{\pi \ell^2}{16}$.
Let us also define $\spheres_m = \cup_{c \in\centers_m} D_{c}$ and $\spheres^* = \cup_{c \in\centers^*} D_{c}$ .
Note that different disks in $\spheres_m$ are pairwise disjoints (except one single point), that $\spheres_m \subset B_{(0,0)}(\rho)$, and that the area of $\spheres_m$ is $|\centers_m|\frac {\pi \ell^2}{16}$.

\begin{figure}[h!]
    \centering
    \def\ray{100pt}
    \begin{subfigure}{0.45\textwidth}
        \begin{tikzpicture}[scale=1]
            \clip[draw] (0,0) circle (\ray) ;
            \fill[red!20] (0pt, 0pt) circle (13pt);
            \fill[red!20] (26pt, 0pt) circle (13pt);
            \fill[red!20] (52pt, 0pt) circle (13pt);
            \fill[red!20] (0pt, -26pt) circle (13pt);
            \fill[red!20] (-26pt, -26pt) circle (13pt);
            \fill[red!20] (0pt, 26pt) circle (13pt);
            \foreach \xcoord in {-52, -26,...,52} {
                \foreach \ycoord in {-52,-26,...,52} {
                    \fill (\xcoord pt, \ycoord pt) circle[radius=1pt];
                    \draw (\xcoord pt, \ycoord pt) circle (13pt);
                }
            }
            \foreach \xcoord in {-78, 78} {
                \foreach \ycoord in {-26,0,26} {
                    \fill (\xcoord pt, \ycoord pt) circle[radius=1pt];
                    \fill (\ycoord pt, \xcoord pt) circle[radius=1pt];
                    \draw (\xcoord pt, \ycoord pt) circle (13pt);
                    \draw (\ycoord pt, \xcoord pt) circle (13pt);
                }
            }
            \draw[->,>=latex, red, xshift=52pt, yshift=52pt] (0,0) -- (30:13pt) node[above left=-6pt, midway] {\small $\ell/4$};
        \end{tikzpicture}
        \caption{Scheme of the general construction of Theorem~\ref{th:lower}, with set $\spheres^*$ colored in red.\label{fig:LB-construct}}
    \end{subfigure}
    \hfill
    \begin{subfigure}{0.45\textwidth}
        \begin{tikzpicture}[scale=1]
            \clip[draw] (0,0) circle (\ray) ;
            \foreach \xcoord in {-52, -26,...,52} {
                \foreach \ycoord in {-52,-26,...,52} {
                    \fill (\xcoord pt, \ycoord pt) circle[radius=1pt];
                }
            }
            \foreach \xcoord in {-78, 78} {
                \foreach \ycoord in {-26,0,26} {
                    \fill (\xcoord pt, \ycoord pt) circle[radius=1pt];
                    \fill (\ycoord pt, \xcoord pt) circle[radius=1pt];
                }
            }
            \draw[color=blue] (0,0) circle (\ray*3/4) ;
            \draw[color=blue, very thick] (-53.04pt, -53.04pt) rectangle (53.04pt, 53.04pt) ;
            
            \draw[->,>=latex, blue, thick] (0,0) -- (30:75pt) node[above, midway] {$\frac{3\rho}4$};
            \draw[->,>=latex, blue, thick] (0,0) -- (53.04pt,0) node[below, midway] {$\frac{3\rho}{4\sqrt2} \geq \frac\rho2$};
            \draw[->, >=latex, blue] (-26pt,0pt) -- (-26pt,26pt) node[right, midway] {$\ell / 2$};
        \end{tikzpicture}
        \caption{Scheme of the proof of Lemma~\ref{lem:grid-included} with $\squarereg$ in blue.\label{fig:LB-nb-spheres}}
    \end{subfigure}
    \caption{Proof of the Lower Bound\label{fig:lower-bound}}
\end{figure}
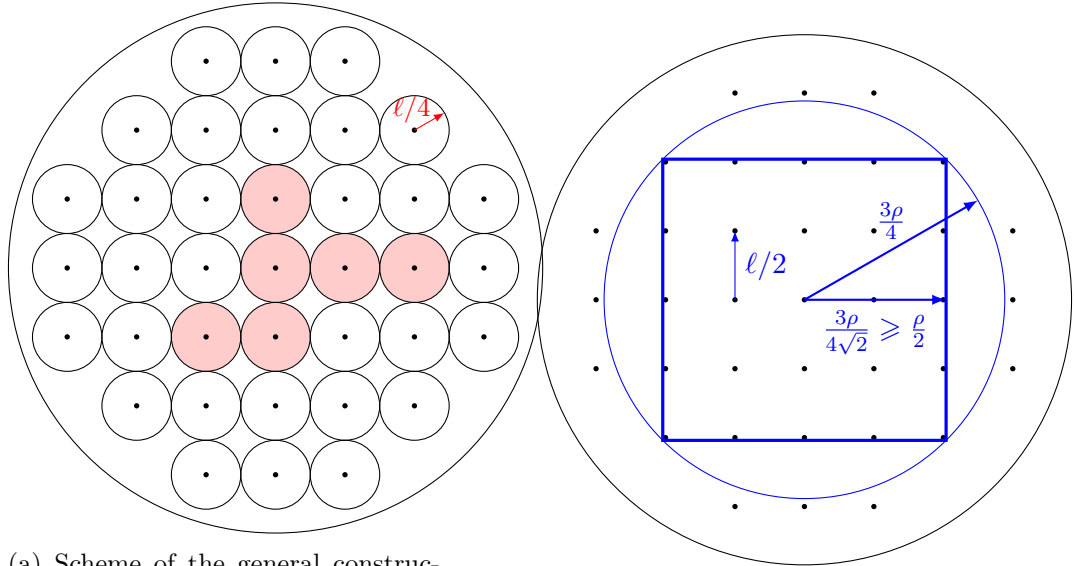

\paragraph{Construction of the set $\P(\A)$.}
Let us first suppose that $n \leq |\centers^*|$.
The construction of the set of initial positions $\P(\A)$ depends on the considered algorithm $\A$.
The process consists in placing exactly one robot per disk $D_c\in\spheres_m$, at position $p_c$.
This construction guarantees that the instance $(\P,s)$ is consistent with the tuple $(\ell,\rho,n)$.
In particular, the set is actually $\ell$-connected according to Lemma~\ref{lem:spheres-connectivity} and because $\centers_m$ is connected.
Given one disk $D_c\in \spheres_m$, the exact localisation $p_c$ is defined as the last position of $S_c$ to be explored by the previously awakened robots, under the execution of the algorithm.
In other words, the algorithm must integrally explore $S_c$ before it discovers the new robot in it.
If $n > |\centers^*|$, then the construction is similar for the first $|\centers^*|-1$ positions.
The initial localization of the remaining robots can be in any arbitrary small disk of radius $\varepsilon$ included in some area of $\spheres_m$ that has not been discovered yet.

\paragraph{Proof of the lower bound $\Omega(\ell^2 \log m)$.}

We denote by $\Awake(t)$ the set of awake robots at time $t$, and given the set of initial positions $\P=\P(\A)$, we denote by $\Ddiscovered{t} \subseteq \spheres_m$ the set of points of $\spheres_m$ that have been discovered at time $t$ or before.
More formally, $\Ddiscovered{t}$ is the set of points $p\in\spheres_m$ such that $\exists t' \leq t, \exists r\in \Awake(t'): |p_r(t')p| \leq 1$.

We denote by $\Adiscovered{t}$ the area of $\Ddiscovered{t}$. 
Finally, let us define, $\forall i \in [0,m], t_i = \inf\{t \geq 0 \mid \Adiscovered{t} \geq i\pi \ell^2/16\}$.
By construction of $\P(\A)$ we have $\forall i, \forall t<t_i, |\Awake(t_i)| \leq i$.
Since robots have a field of view of radius $1$, they discover an area of amplitude $2$ along an unitary move, which mean that an awake robot exploring during $t'$ units of time can discover an area of at most $2t'$.
Therefore we have $\forall i, \forall t\geq t_i, \Adiscovered{t} - \Adiscovered{t_i} \leq 2(t - t_i) |\Awake(t)|$. 
Furthermore, $\Adiscovered{t_{i+1}} - \Adiscovered{t_i} = \frac {\pi \ell^2}{16}$.
By having $t$ tend to $t_{i+1}$ by inferior values, we obtain:
$\forall i\geq 0, \Adiscovered{t_{i+1}} - \Adiscovered{t_i} \leq 2(t_{i+1} - t_i) \times (i+1)$.

Therefore by adding this telescopic sum, we obtain:
\[
    t_m \geq \frac{\pi \ell^2}{32} \sum_{i=1}^m i \geq \frac{\pi \ell^2}{32} \ln (m+1) \in \Omega(\ell^2 \clog m)
\]

\paragraph{Conclusion.}
Since $n\geq \frac \rho \ell$ and $|\centers^*|\geq \rho^2/\ell^2$, and $m = \min (n, |\centers^*|)$, we have obtained a lower bound of $\Omega(\ell^2 \clog \frac \rho \ell)$.
Furthermore, by construction of $\centers^m$, and by definition of $\P$, robots have explore all points of $D_{(0,\floor{\rho/\ell}\frac\ell2)}$, which means that there exists a path from $(0,0)$ to a point at distance $1$ of $(0,\floor{\rho/\ell}\frac\ell2)$.
Such a path has a length at least $\floor{\rho/\ell}\frac\ell2 -1\geq \frac\rho{2\ell}\frac\ell2 -1 \geq \rho/4 -1 \in \Omega(\rho)$.
We therefore have shown that the makespan of $\A$ on that set of point is $\Omega(\rho + \ell^2 \clog \frac \rho \ell)$.

\paragraph{Intermediate Lemmas.}

\begin{lemma}
    \label{lem:grid-included}
    $|\centers| \geq 1 + \rho^2/\ell^2$.
\end{lemma}
\begin{proof}
    Let us consider $\squarereg$ the square included in the disk with center $(0,0)$ and radius $\frac {3\rho}4$, itself included in the disk with same center and with radius $\rho-\frac \ell4$, as pictured in Figure~\ref{fig:LB-nb-spheres}.
    Let us count the number of points of $\centers$ included in $\squarereg$.
    The semi-width of this square is $\frac{3\rho}{4\sqrt 2} \geq \frac \rho2$.
    Therefore, the number of points of $\centers$ on one vertical line of $\squarereg$ is $1 + 2\floor{\frac{\rho/2}{\ell/2}} = 1+2\floor{\frac \rho \ell}$.
    Since $\rho\geq \ell$, we have $\floor{\frac \rho \ell} \geq \frac 12 \frac \rho \ell$, and thus the number of points of $\centers$ in $\squarereg$ is at least $(1 + \frac \rho \ell)^2 \geq 1 + \rho^2/\ell^2$.
\end{proof}

\begin{lemma}
    \label{lem:spheres-connectivity}
    For any adjacent points $c,c' \in \centers_m$, for any two points $(p,p') \in D_c \times D_{c'}$, we have $|pp'| \leq \ell$.
\end{lemma}
\begin{proof}
    Both $D_c$ and $D_{c'}$ are included in the disk of center $\frac{c+c'}2$ and radius $\frac \ell2$.
\end{proof}

\subsection{About energy constraint}
\label{sec:impossibility}
Let us recall and prove Theorem~\ref{th:impossibility}:
\impossibility*
\begin{proof}
    We construct the $n$-point set in a similar way as in the proof of Theorem~\ref{th:lower}.
    Similarly as before, let us first present the situation where $n=1$.
    The position of robot $r_1$ in v depends of the behavior of algorithm \A. More precisely, it is the last position of $B_{(0,0)}(\ell)$ to be discovered by $s$.
    Therefore, $s$ can wake up robot $r_1$ only after it has discovered the entire disk $B_{(0,0)}(\ell)$, with area $\pi \ell^2$.
    At $t=0$, the discovered area is $\pi$, and corresponds to the disk of radius $1$ around $s$.
    During a move of amplitude $\delta$, robot $s$ discover a new area of at most $2\delta$, which mean that to discover the entire disk, $s$ has to move with a total amplitude of at least $\frac{\pi\ell^2 - \pi}2$, which concludes the proof.
    If $n>1$, then we can either place all the sleeping robots at the exact same position, either consider a small enough disk with radius $\varepsilon$ which is not discovered by $s$ before it has consumed its entire energy.
\end{proof}

\subsection{With energy constraint}\label{sec:lower-nrj}
Let us recall Theorem~\ref{th:lower-nrj}:
\lowernrj*

We first show how to define a $n$-point set $\P$ and a source $s$ with a prescribed eccentricity \ecc such that $\ellstar \leq \ell$, $\rhostar \leq \rho$ and $\ecc_\ell = \ecc$. We prove in Lemma~\ref{lem:path-constructivity} that the construction is valid.
We then prove Theorem~\ref{th:lower-nrj} using this construction and Theorem~\ref{th:lower}.

Let us briefly explain how we construct $(\P,s)$.
The source $s$ has position $\pos[s] = (0,0)$.
The main idea is to define the positions of points \P along a rectilinear path $\Pi$, that is a path that consists only of horizontal and vertical segments, defined such that any point on a horizontal segment is space out of at least $B + 1$ from any other horizontal segment.
Positions of $\P$ are spread over $\Pi$ such that the $\ell$-disk graph of $\P \cap \pos[s]$ is indeed connected.

Our construction of $\Pi$ is such its length determines the $\ell$-eccentricity of $(\P,s)$.
The length of $\Pi$ should match the prescribed $\ecc$, which yields some constraints over the values of $\ecc$, depending on the values of $\rho, n$ and $B$. 

\paragraph{Path definition} The rectilinear path $\Pi$ is defined by vertical segments of length $V = B+1$ and horizontal segments of length $H =  \rho/\sqrt{2}$.
Let $J = \floor{\ecc/(H + V)}$ be the number of vertical segments. 
We define $\Pi$ by the points $u_0, v_0, v_1, u_1, u_2, v_2,..., u_{J-1}, v_{J-1}, v_{J}, u_{J}$.
There is a vertical segment between $[v_jv_{j+1}]$ or $[u_ju_{j+1}]$ and a horizontal segment between $[u_jv_j]$. 
The position of points are more formally defined by the following coordinates.
\begin{equation*}
    \forall j \in [0, J]: \left|
    \begin{aligned}
      u_j &= (0, j(B+1))\\
      v_j &= (\frac{\rho}{\sqrt{2}}, j (B+1))
    \end{aligned}
    \right.
\end{equation*}
For $j \in [0, J[$ the $j$-th \emph{section} of $\Pi$ is a sequence of a horizontal segment $[u_jv_j]$ and a vertical segment, either $[v_jv_{j+1}]$ for even $j$ or $[u_ju{j+1}]$ if $j$ is odd.

If no point of $\Pi$ is at distance $\rho$ from $\pos[s]$ then we need to define an additional segment line between $v_0 = (\rho/\sqrt{2}, 0)$ and a point $w_0 = (\rho, 0)$ to spread out some points of $\P$ on $[v_0w_0]$ until $\rhostar = \rho$.
Conversely, note that $\Pi$ should fit in a square $S_{\rho}$ of width $\rho/\sqrt{2}$ whose bottom left corner is $\pos[s]$. 
Thus the sum of the length of vertical segments $JV$ can not exceed $\rho/\sqrt{2}$.
\begin{align}
    \frac{\ecc}{H + V}  &\leq \frac{H}{V} \\
    \ecc &\leq \frac{H^2 + V}{V}\\
    \ecc &\leq \frac{\rho^2}{2(B+1)} + 1
\end{align}

\paragraph{Placements of $\P$.} 
Let us set the first position $\pos[t]$ such that the sub-path of $\Pi$ between $\pos[s]$ and $\pos[t]$ has length exactly $\ecc$.
This point lies on the the $J$-th section of $\Pi$.
Other points of $\P$ are positioned on the subpath from $\pos[s]$ to $\pos[t]$ and on $[\pos[s]w_0]$ if needed.
The $\ell$-eccentricity of the resulting $(\P,s)$ must be $\ecc$.
This requires to carefully assign positions around a corner $\widehat{u_{j}v_jv_{j+1}}$ or $\widehat{v_ju_ju_{j+1}}$ in order to avoid any shortcut in the $\ell$-disk graph between horizontal and vertical segment.

We start by placing a subset $\P_1 \subset \P/\{\pos[t]\}$ of size $2(J-1)$ on the extremities of each segment except for $u_0$ where $\pos[s]$ lies and $v_J$ (\textit{resp.} $u_J$) if $J$ is even (\textit{resp.} odd).
This ensure that no points are positioned beyond $\pos[t]$.
Remaining positions are placed such that (1) there is no point at distance strictly less than $\ell$ from a point of $\P_1$ and (2) the resulting $(\P, s)$ has connectivity threshold $\ell$.

Note that this requirement about the connectivity threshold constrains the eccentricity of $(\P, s)$ as the number $n$ may not be always large enough to cover $\Pi$.
There must be at least $\ecc/\ell$ positions of $\P$ on $\Pi$ and $\rho(1 - 1/\sqrt{2}) \leqslant \rho/3\ell$ positions on $[v_0w_0]$.
Therefore the number of positions $n$ is such that $n \geqslant \ecc/\ell + \rho/3\ell$.
This is true whenever $\ecc \leqslant n\ell - \rho/3$.

\begin{lemma}
    \label{lem:path-constructivity}
    The construction of $(\P,s)$ such that $\rhostar = \rho$, $\ellstar = \ell$ and $\ecc_\ell = \ecc$.
\end{lemma}

\begin{proof}
    Firstly we show $\rhostar = \rho$. 
    Since the path $\Pi$ is entirely contained in the square $S_{\rho}$, the farthest point from $\pos[s]$ to a point on $\Pi$ is the up right corner of $S_{\rho}$. 
    Its distance from $\pos[s]$ is equal to the diagonal of $S_{\rho}$ which is exactly $\rho$.
    If no points of $\Pi$ does not reach that corner, then we added the segment $[v_0w_0]$ which ensures that there is at least one point at distance $\rho$ from $\pos[s]$.
    
    We show that $\ecc_\ell = \ecc$.
    By Theorem~\ref{th:impossibility}, $B > \ell$, so the eccentricity of the $\ell$-disk graph of $\P$ is the maximum between the length of $\Pi$ and the length of $[\pos[s], w_0]$ which is $\rho$.
    Since by assumption $\ecc \geqslant \rho$, we get that the eccentricity of the $\ell$-disk graph of $\P$ is the length of $\Pi$.
    Let us compute the length of $\Pi$. By definition it is 

    \begin{align}
        J(H + V)= \floor{\frac{\ecc}{H + V}}(H+V)= \ecc
    \end{align}
     
    By construction we directly have $\ellstar = \ell$.
\end{proof}

Let us prove Theorem~\ref{th:lower-nrj}.
\begin{proof}
    Let us start by showing that the makespan of the instance $(\P, s)$ constructed above is $\Omega(\ecc)$.
    Any wake-up strategy can not be quicker than the eccentricity of $s$ in the $B$-disk graph of $(\P,s)$.
    In the $B$-disk graph any path from $\pos[s]$ to $\pos[t]$ goes entirely through section.
    This is because for $j \in [0, J[$ any point of $[u_jv_j]$ is at distance $B+1$ from any point of $[u_{j+1},v_{j+1}]$, so the path goes trough the points of the vertical segment $[u_j, u_{j+1}]$ or $[v_j, v_{j+1}]$ to reach $[u_{j+1},v_{j+1}]$.

    Such a path has a length longer than the sum of $|u_jv_j|$, that is $JH = \floor{\frac{\ecc}{H + V}} H$. 
    
    We have several cases depending the value of $J$:
    \begin{itemize}
        \item If $J=0$, the $\ell$-eccentricity of $(\P,s)$ is $[\pos[s], w_0] = \rho$ and we have $\rho \leq \ecc \leq H + V \leq 2\rho/\sqrt{2}$ therefore $\ecc_\ell = \rho = \Theta(\ecc)$
        \item If $J=1$, then $\rho \leq \ecc \leq 2(H + V) \leq 2\sqrt{2}\rho$ and $\ecc_\ell = \Theta(\ecc)$
        \item If $J \geq 2$, then $\ecc \geq 2(H+V)$
        \begin{align*}
            \floor{\frac{\ecc}{H + V}} H &\geq \frac{\ecc H}{H + V} - H\\
            &\geq \frac{H(\ecc - H - V)}{H+V} \geq \frac{\ecc}{4}
        \end{align*}
    \end{itemize}
    
    We conclude the proof by recalling that the lower bound with unconstrained energy naturally apply to dFTP with constrained energy
    (See Theorem~\ref{th:lower}), and by noticing that $\rho/\ell \leq \ecc_\ell/\ell \leq 12\rho^2/\ell^2$.
    Therefore the makespan of \A is $\Theta(\ecc_\ell + \ell^2\clog(\ecc_\ell/\ell))$.
\end{proof}

\bibliographystyle{alpha}
\bibliography{a.bib}

\end{document}